\documentclass[aoas,preprint]{imsart}

\RequirePackage[OT1]{fontenc}
\RequirePackage{natbib}
\RequirePackage{amsthm,amsmath}

\usepackage[a4paper, top=3cm, bottom=3cm]{geometry}
\usepackage[latin1]{inputenc}

\usepackage{tocloft}
\usepackage[usenames]{color}
\usepackage{amsmath,amsfonts,amssymb,amsthm}
\usepackage{graphicx}
\usepackage{algorithm}
\usepackage{algpseudocode}

\usepackage{paralist}

\usepackage{bm}

\usepackage{subfig}

\ifpdf
\usepackage[colorlinks,urlcolor=blue,citecolor=blue,linkcolor=blue]{hyperref}
\else
\newcommand{\href}[2]{{#2}}
\usepackage{url}
\fi

\ifpdf
\newcommand{\Sec}[1]{\hyperref[sec:#1]{Section~\ref*{sec:#1}}} 
\newcommand{\App}[1]{\hyperref[sec:#1]{Appendix~\ref*{sec:#1}}} 
\newcommand{\Supp}[1]{\hyperref[sec:#1]{Supplement~\ref*{sec:#1}}} 
\newcommand{\Eqn}[1]{\hyperref[eq:#1]{{\rm (\ref*{eq:#1})}}} 
\newcommand{\Part}[1]{\hyperref[part:#1]{(\ref*{part:#1})}} 
\newcommand{\Fig}[1]{\hyperref[fig:#1]{Figure~\ref*{fig:#1}}} 
\newcommand{\Tab}[1]{\hyperref[tab:#1]{Table~\ref*{tab:#1}}} 
\newcommand{\Thm}[1]{\hyperref[thm:#1]{Theorem~\ref*{thm:#1}}} 
\newcommand{\Lem}[1]{\hyperref[lem:#1]{Lemma~\ref*{lem:#1}}} 
\newcommand{\Prop}[1]{\hyperref[prop:#1]{Proposition~\ref*{prop:#1}}} 
\newcommand{\Cor}[1]{\hyperref[cor:#1]{Corollary~\ref*{cor:#1}}} 
\newcommand{\Def}[1]{\hyperref[def:#1]{Definition~\ref*{def:#1}}} 
\newcommand{\Alg}[1]{\hyperref[alg:#1]{Algorithm~\ref*{alg:#1}}} 
\newcommand{\Ex}[1]{\hyperref[ex:#1]{Example~\ref*{ex:#1}}} 
\newcommand{\As}[1]{\hyperref[as:#1]{Assumption~{\rm\ref*{as:#1}}}} 
\newcommand{\Reg}[1]{\hyperref[as:#1]{Condition~\ref*{reg:#1}}} 
\newcommand{\AlgLine}[2]{\hyperref[alg:#1]{line~\ref*{line:#2} of Algorithm~\ref*{alg:#1}}}
\newcommand{\AlgLines}[3]{\hyperref[alg:#1]{lines~\ref*{line:#2}--\ref*{line:#3} of Algorithm~\ref*{alg:#1}}}
\else
\newcommand{\Sec}[1]{{Section~\ref{sec:#1}}} 
\newcommand{\App}[1]{{Appendix~\ref{sec:#1}}} 
\newcommand{\Supp}[1]{{Supplement~\ref{sec:#1}}} 
\newcommand{\Eqn}[1]{{(\ref{eq:#1})}} 
\newcommand{\Part}[1]{{(\ref{part:#1})}} 
\newcommand{\Fig}[1]{{Figure~\ref{fig:#1}}} 
\newcommand{\Tab}[1]{{Table~\ref{tab:#1}}} 
\newcommand{\Thm}[1]{{Theorem~\ref{thm:#1}}} 
\newcommand{\Lem}[1]{{Lemma~\ref{lem:#1}}} 
\newcommand{\Prop}[1]{{Proposition~\ref{prop:#1}}} 
\newcommand{\Cor}[1]{{Corollary~\ref{cor:#1}}} 
\newcommand{\Def}[1]{{Definition~\ref{def:#1}}} 
\newcommand{\Alg}[1]{{Algorithm~\ref{alg:#1}}} 
\newcommand{\Ex}[1]{{Example~\ref{ex:#1}}} 
\newcommand{\Reg}[1]{{R~\ref*{reg:#1}}} 
\fi

\newcommand{\Real}{\mathbb{R}}
\newcommand{\dom}{{\bf dom}\,}

\newcommand{\Tra}{^{\sf T}} 
\def\vec{\mathop{\rm vec}\nolimits}

\newcommand{\V}[1]{{\bm{\mathbf{\MakeLowercase{#1}}}}} 
\newcommand{\VE}[2]{\MakeLowercase{#1}_{#2}} 
\newcommand{\Vn}[2]{\V{#1}^{(#2)}} 
\newcommand{\Vtilde}[1]{{\bm{\tilde \mathbf{\MakeLowercase{#1}}}}} 

\newcommand{\M}[1]{{\bm{\mathbf{\MakeUppercase{#1}}}}} 
\newcommand{\ME}[2]{\MakeLowercase{#1}_{#2}} 
\newcommand{\Mtilde}[1]{{\bm{\tilde \mathbf{\MakeUppercase{#1}}}}} 
\newcommand{\Mbar}[1]{{\bm{\bar \mathbf{\MakeUppercase{#1}}}}} 
\newcommand{\Mn}[2]{\M{#1}^{(#2)}} 

\newcommand{\Kron}{\otimes} 

\newcommand{\prox}{\mathop{\rm prox}\nolimits}

\newcommand{\amp}{\mathop{\:\:\,}\nolimits}

\newtheorem{assumption}{Assumption}[section]
\newtheorem{proposition}{Proposition}[section]
\newtheorem{theorem}{Theorem}[section]

\startlocaldefs
\numberwithin{equation}{section}
\theoremstyle{plain}

\endlocaldefs

\begin{document}

\begin{frontmatter}
\title{Convex Biclustering}
\runtitle{Convex Biclustering}

\begin{aug}
\author{\fnms{Eric C.} \snm{Chi}\ead[label=e1]{eric$\_$chi@ncsu.edu}},
\author{\fnms{Genevera I.} \snm{Allen}\ead[label=e2]{gallen@rice.edu}}
\and
\author{\fnms{Richard G.} \snm{Baraniuk}
\ead[label=e3]{richb@rice.edu}
}

\runauthor{E. Chi et al.}

\affiliation{North Carolina State University and Rice University}

\address{E. C. Chi\\
Department of Statistics \\
North Carolina State University \\
Raleigh, NC 27695\\
USA \\
\printead{e1}}

\address{G. I. Allen\\
Department of Statistics\\
Rice University\\
Houston, TX 77005\\
USA\\
\printead{e2}}

\address{R. G. Baraniuk\\
Department of Electrical\\
and Computer Engineering\\
Rice University\\
Houston, TX 77005\\
USA \\
\printead{e3}}
\end{aug}

\begin{abstract}
In the biclustering problem, we seek to simultaneously group observations and features. While biclustering has applications in a wide array of domains, ranging from text mining to collaborative filtering, the problem of 
identifying structure in high dimensional genomic data motivates this work. In this context, biclustering enables us to identify subsets of genes that are co-expressed only within a subset of experimental conditions. We present a convex formulation of the biclustering problem that possesses a unique global minimizer and an iterative algorithm, COBRA, that is guaranteed to identify it. Our approach generates an entire solution path of possible biclusters as a single tuning parameter is varied. We also show how to reduce the problem of selecting this tuning parameter to solving a trivial modification of the convex biclustering problem. The key contributions of our work are its simplicity, interpretability, and algorithmic guarantees - features that arguably are lacking in the current alternative algorithms. We demonstrate the advantages of our approach, which includes stably and reproducibly identifying biclusterings, on simulated and real microarray data.
\end{abstract}

\begin{keyword}
\kwd{Clustering}
\kwd{Convex Optimization}
\kwd{Fused Lasso}
\kwd{Gene Expression}
\kwd{Reproducible Research}
\kwd{Structured Sparsity}
\end{keyword}

\end{frontmatter}

\section{Introduction}
\label{sec:introduction}

In the biclustering problem, we seek to simultaneously group observations (columns) and features (rows) in a data matrix. Such data is sometimes described as two-way, or transposable, to put the rows and columns on equal footing and to emphasize the desire to uncover structure in both the row and column variables. Biclustering is used for visualization and exploratory analysis in a wide array of domains. For example, in text mining, biclustering can identify subgroups of documents with similar properties with respect to a subgroup of words \citep{Dhi2001}. In collaborative filtering, it can be used to identify subgroups of customers with similar preferences for a subset of products \citep{HofPuz1999}. Comprehensive reviews of biclustering methods can be found in \citep{MadOli2004,TanSha2005,BusPro2008}. 

In this work, we focus on biclustering to identify patterns in high dimensional cancer genomic data.  While a cancer, such as breast cancer, may present clinically as a homogenous disease, it typically consists of several distinct subtypes at the molecular level. A fundamental goal of cancer research is the identification of subtypes of cancerous tumors that have similar molecular profiles and the genes that characterize each of the subtypes. Identifying these patterns is the first step towards developing personalized treatment strategies targeted to a patient's particular cancer subtype. 

Subtype discovery can be posed as a biclustering problem in which gene expression data is partitioned into a checkerboard-like pattern that highlights the associations between groups of patients and groups of genes that distinguish those patients. Biclustering has had some notable successes in subtype discovery. For instance,
biclustering breast cancer data has identified sets of genes whose expression levels segregated patients into five subtypes with distinct survival outcomes \citep{SPerTib2001}.  These subtypes have been reproduced in numerous studies \citep{STibPar2003}. Encouraged by these results, scientists have searched for molecular subtypes in other cancers, such as ovarian cancer \citep{TotTinGeo2008}. Unfortunately, the findings in many of these additional studies have not been as reproducable as those identified by \cite{SPerTib2001}. The failure to reproduce these other results may reflect a genuine absence of biologically meaningful groupings. But another possibility may be related to issues inherent in the computational methods currently used to identify biclusters.

While numerous methods for biclustering genomic data have been proposed \citep{MadOli2004,BusPro2008}, the most popular approach to biclustering cancer genomics data is the clustered dendrogram.  This method performs hierarchical clustering \citep{HasTib2009} on the patients (columns) as well as on the genes (rows).  The matrix is then re-ordered according to these separate row and column dendrograms and visualized as a heatmap with dendrograms plotted alongside the row and column axes.  \Fig{lung500_hclust} illustrates an example of a clustered dendrogram of expression data from a lung cancer study. The data consists of the expression levels of 500 genes across 56 individuals, a subset of the data studied in \cite{LeeSheHua2010}. Subjects belong to one of four subgroups: Normal, Carcinoid, Colon, or Small Cell.
\begin{figure}
        \centering
\begin{tabular}{cc}
\subfloat[Raw Data]{\includegraphics[scale=0.4125]{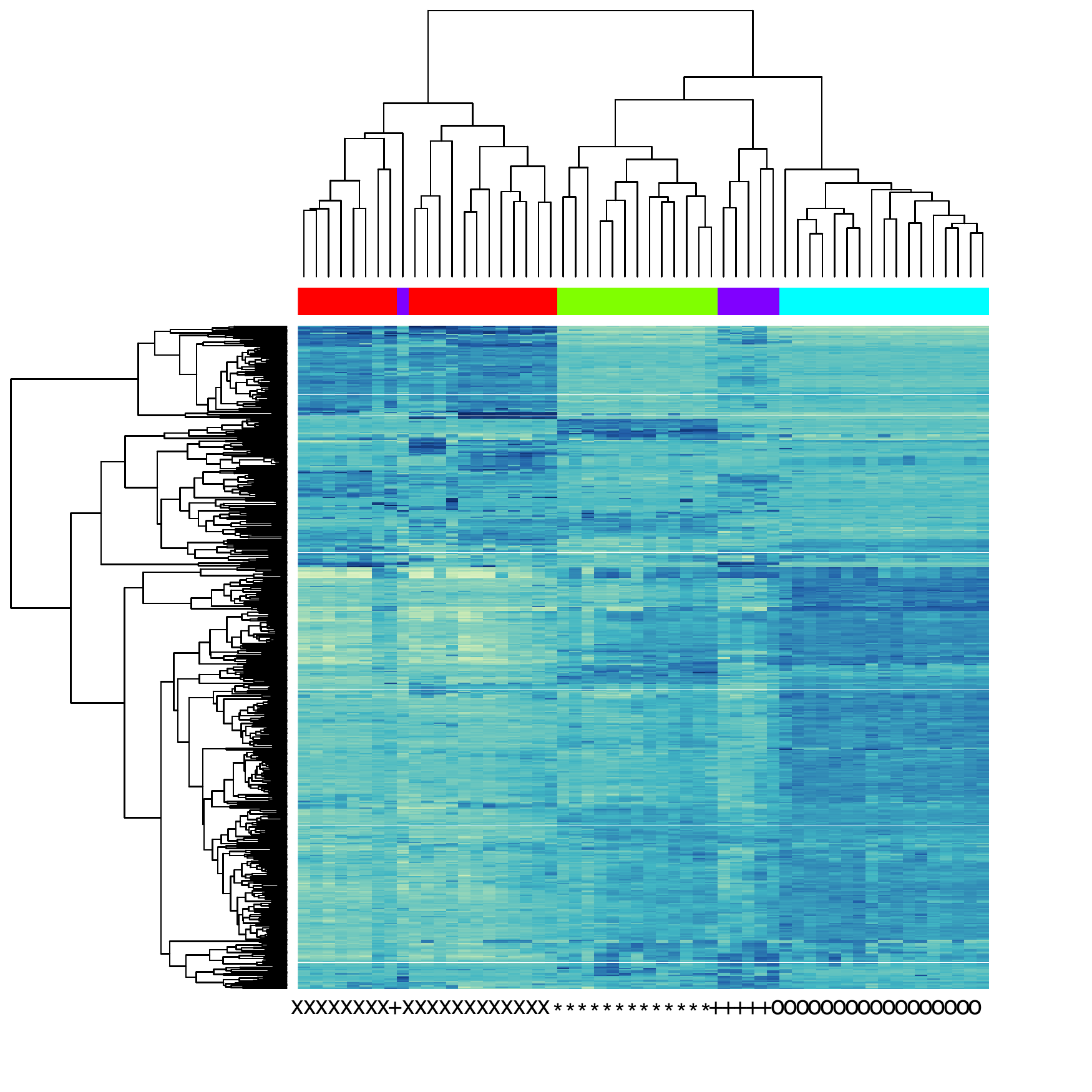}
\label{fig:lung500_hclust}} 
   & \subfloat[COBRA Smoothed Estimate]{\includegraphics[scale=0.4125]{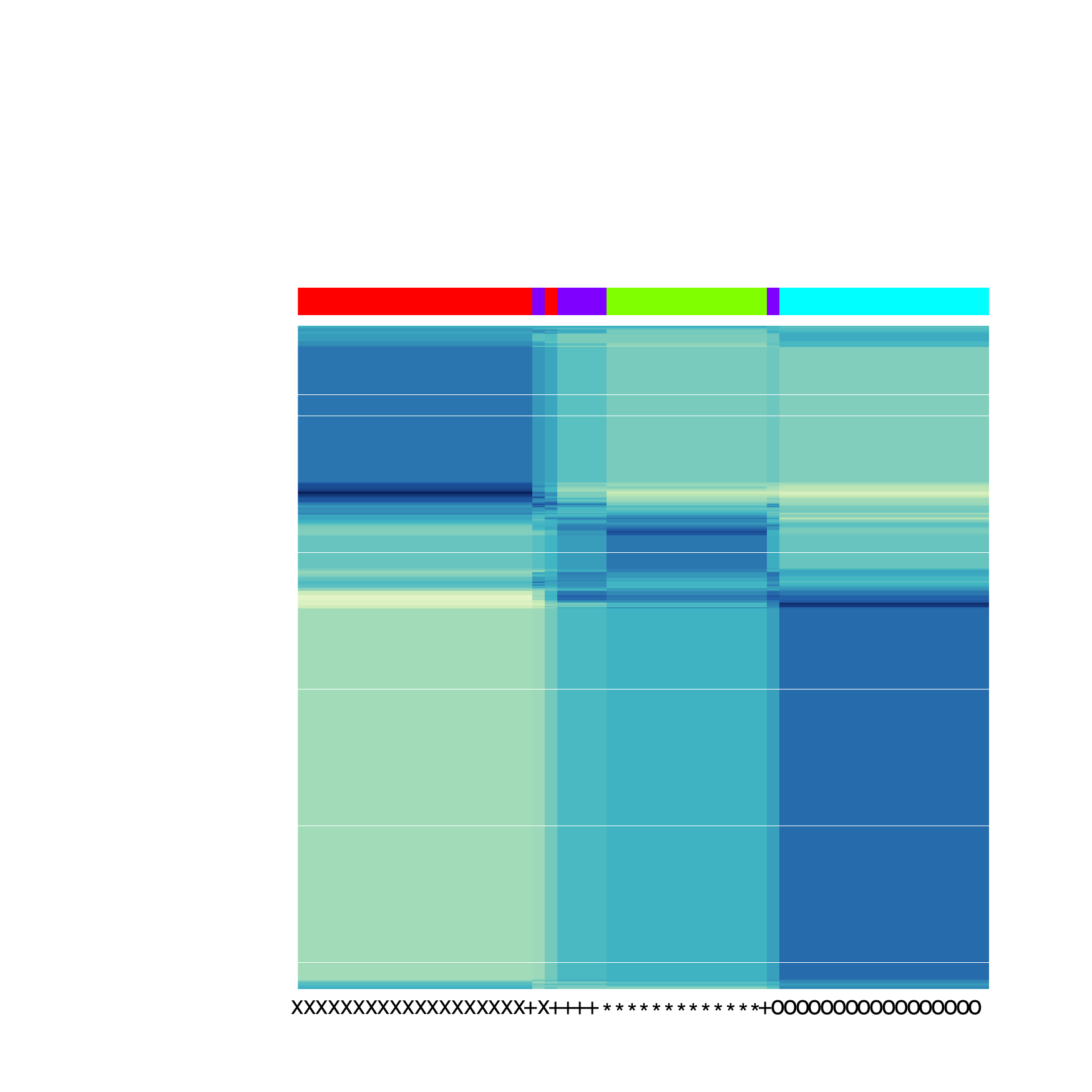}
\label{fig:lung500_cba}}\\ \vspace{2mm}
\end{tabular}
\caption{Heatmaps of the expression of 500 genes (rows) across 56 subjects (columns). \Fig{lung500_hclust} depicts the clustered dendrogram applied to the raw data; \Fig{lung500_cba} depicts COBRA smoothed data after reordering the columns and rows via the seriation package \citep{Hahsler2008}. Subjects belong to one of four subgroups: Normal (o), Carcinoid (x), Colon (*), and Small Cell (+).}

\end{figure}
This simple strategy seems to reasonably recover the clinical diagnoses and identify sets of genes whose dysregularization characterizes the subtypes.

As an algorithmic procedure, however, the clustered dendrogram has two characteristics that make it less than ideal for generating reproducible results. Dendrograms are constructed by greedily fusing observations (features) to decrease some criterion. Consequently the algorithm may return a biclustering that is only locally optimal with respect to the criterion. Since solutions may vary depending on how the algorithm is initialized, such procedures tend to be run from multiple initializations, but even then there is no guarantee that a global optimum will be reached. The algorithm is also not stable in the sense that small perturbations in the data can lead to large changes in the clustering assignments. 

More sophisticated methods have been proposed for biclustering, some based on the singular value decomposition (SVD) of the data matrix  \citep{LazOwe2002,BerIhm2003,TurBai2005,WitTib2009,LeeSheHua2010,SilKaiKop2011}, and others based on graph cuts \citep{Dhi2001,KluBasGer2003}. Some approaches are similar in spirit to the clustered dendrogram and directly cluster the rows and columns \citep{CoiGav2011,TanWit2013}. While these methods may provide worthwhile improvements in empirical performance, none of them address the two fundamental issues that dog the clustered dendrogram. Moreover, scientists may shy from using many of these methods, since their outputs are typically not as easy to visualize as compared to the simple clustered dendrogram. From a reproducible research perspective, a biclustering method should (i) give the same, ideally unique, answer regardless of how the algorithm is initialized, and (ii) be stable with respect perturbations in the data.

In this paper, we pose biclustering as a convex optimization problem and introduce a novel Convex BiclusteRing Algorithm (COBRA) for iteratively solving it. COBRA outputs results that retain the simple interpretability and visualization of the clustered dendrogram and also possesses several key advantages over existing techniques: 
\begin{inparaenum}[(i)]
\item {\bf Stability and Uniqueness:} COBRA produces the unique global minimizer to a convex program and this minimizer is continuous in the data. This means that COBRA  always maps the data to a single biclustering assignment, and this solution is stable.
\item {\bf Simplicity:} COBRA employs a single tuning parameter that controls the number of biclusters.
\item {\bf Data Adaptivity:} COBRA admits a simple and principled data adaptive procedure for choosing the tuning parameter that involves solving a convex matrix completion problem.
\end{inparaenum}

Returning to our motivating lung cancer example, \Fig{lung500_cba} illustrates the results of COBRA with the tuning parameter selected according to our data adaptive procedure. After reordering the columns and rows via the seriation package \citep{Hahsler2008}, with the `TSP' option, a clearer biclustering patterns emerge. 

\section{A Convex Formulation of Biclustering}
\label{sec:formulation}

Our goal is to identify the groups of rows and groups of columns in a data matrix that are associated with each other. As seen in \Fig{lung500_cba}, when rows and columns are reordered according to their groupings, a checkerboard pattern emerges, namely the elements of the matrix partitions defined by row and column groups tend to display a uniform intensity.

We now describe a probabilistic model that can generate the observed checkerboard pattern. Our data, $\M{X} \in \Real^{p \times n}$, consists of $n$ samples drawn from a $p$-dimensional feature space. Suppose that the latent checkerboard structure is defined by $R$ feature groups and $C$ observation groups. If the $ij$th entry in $\M{X}$ belongs to the cluster defined by the $r$th feature group and $c$th observation group, then we assume that the observed value $\VE{x}{ij}$ is given by
$\ME{x}{ij} = \ME{\mu}{0} + \ME{\mu}{rc} + \VE{\varepsilon}{ij}$,
where $\ME{\mu}{0}$ is a baseline or grand mean shared by all entries of the data matrix, $\ME{\mu}{rc}$ is the mean of the cluster defined by the $r$th row partition and $c$th column partition, and $\VE{\varepsilon}{ij}$ are iid\@ $N(0,\sigma^2)$ for some $\sigma^2 > 0$. To ensure identifiability of the mean parameters, we assume that $\ME{\mu}{0} = 0$, which can be achieved by removing the grand sample mean from the data matrix $\M{X}$.

This biclustering model corresponds to a checkerboard mean model \citep{MadOli2004}. This model is most similar to that assumed by \cite{TanWit2013} who propose methods to estimate a checkerboard-like structure with some of the bicluster mean entries being sparse. The checkerboard model is exhaustive in that each matrix element is assigned to one bicluster. This is in contrast to other biclustering models that identify potentially overlapping row and column subsets that are not exhaustive; these are typically estimated using SVD-like methods \citep{LazOwe2002,BerIhm2003,TurBai2005,WitTib2009,LeeSheHua2010,SilKaiKop2011} or methods to find hot-spots \citep{ShaWeiNob2009}. 

Estimating the checkerboard model parameters consists of finding the partitions and the mean values of each partition. Estimating $\ME{\mu}{rc}$, given feature and observation clusterings, is trivial. Let $\mathcal{R}$ and $\mathcal{C}$ denote the indices of the $r$th row partition and $c$th column partition. The maximum likelihood estimate of $\ME{\mu}{rc}$ is simply the sample mean of the entries of $\M{X}$ over the indices defined by $\mathcal{R}$ and $\mathcal{C}$, namely
$\ME{\mu}{rc} = \frac{1}{\lvert \mathcal{R} \rvert\lvert \mathcal{C} \rvert}\sum_{i \in \mathcal{R}, j \in \mathcal{C}} \ME{X}{ij}.$

In contrast, estimating the row and column partitions, is a combinatorially hard problem and characterizes the main objective of biclustering. This task is akin to best subset selection in regression problems \citep{HasTib2009}. However, just as the best subset selection problem has been successfully attacked by solving a convex surrogate problem, namely the Lasso \citep{Tib1996}, we will develop a convex relaxation of the combinatorally hard problem of selecting the row and column partitions.

We propose to identify the partitions by minimizing the following convex criterion
\begin{eqnarray}
F_{\gamma}(\M{U}) & = & \frac{1}{2} \lVert \M{x} - \M{U} \rVert_{\text{F}}^2 + \gamma \underbrace{\left [\Omega_{\M{w}}(\M{U}) + \Omega_{\Mtilde{W}}(\M{U}\Tra) \right ]}_{J(\M{U})},
\label{eq:biclust_objective_function}
\end{eqnarray}
where $\Omega_{\M{w}}(\M{U}) = \sum_{i<j}w_{ij} \|\M{U}_{\cdot i}-\M{U}_{\cdot j} \rVert_2$,
and $\M{U}_{\cdot i}$ ($\M{U}_{i \cdot})$ denotes the $i$th column (row) of the matrix $\M{U}$. 

We have posed the biclustering problem as a penalized regression problem, where
the matrix $\M{U} \in \Real^{p \times n}$ is our estimate of the means matrix $\M{\mu}$. The quadratic term quantifies how well $\M{U}$ approximates $\M{X}$. The regularization term $J(\M{U})$ penalizes deviations away from a checkerboard pattern. The parameter $\gamma \geq 0$ tunes the tradeoff between the two terms. The parameters $w_{ij} = w_{ji}$ and $\tilde{w}_{ij} = \tilde{w}_{ji}$ are non-negative weights that will be explained shortly.

The penalty term $J(\M{U})$ is closely related to other sparsity inducing penalties. When only the rows or columns are being clustered, minimizing the objective function in \Eqn{biclust_objective_function} corresponds to solving a convex clustering problem \citep{PelDeSuy2005,HocVerBac2011,LinOhlLju2011} under an $\ell_2$-norm fusion penalty. 
The convex clustering problem in turn can be seen as a generalization of the Fused Lasso \citep{TibSauRos2005}. 
When the $\ell_1$-norm is used in place of the $\ell_2$-norm in $\Omega_{\M{W}}(\M{U})$, we recover a special case of the general Fused Lasso \citep{TibTay2011}.
Other norms can also be employed in our framework; see for example \cite{ChiLan2015}. In this paper, we restrict ourselves to the $\ell_2$-norm since it is rotationally invariant. In general, we do not want a procedure whose biclustering output may non-trivially change when the coordinate representation of the data is trivially changed.

To understand how the regularization term $J(\M{U})$ steers solutions toward a checkerboard pattern, consider the effects of $\Omega_{\M{W}}(\M{U})$ and $\Omega_{\Mtilde{W}}(\M{U}\Tra)$ separately. Suppose $J(\M{U}) = \Omega_{\M{W}}(\M{U})$. The $i$th column $\M{U}_{\cdot i}$ of the matrix $\M{U}$ can be viewed as a cluster center or centroid attached to the $i$th column $\M{x}_{\cdot i}$ of the data matrix $\M{X}$. When $\gamma=0$, the minimum is attained when $\M{U}_{\cdot i}=\M{x}_{\cdot i}$, and each column occupies a unique column cluster.
As $\gamma$ increases, the cluster centroids are shrunk together and in fact begin to coalesce. Two columns $\M{x}_{\cdot i}$ and $\M{x}_{\cdot j}$ are assigned to the same column partition if $\M{U}_{\cdot i} = \M{U}_{\cdot j}$. We will prove later that for sufficiently large $\gamma$, all columns coalesce into a single cluster. Similarly, suppose $J(\M{U}) = \Omega_{\Mtilde{W}}(\M{U}\Tra)$ and view the rows of $\M{U}$ as the cluster centroids attached to the rows of $\M{X}$. As $\gamma$ increases, the row centroids will begin to coalesce, and we likewise say that the $i$th and $j$th rows of $\M{X}$ belong to the same row partition if their centroid estimates $\M{U}_{i \cdot}$ and $\M{U}_{j \cdot}$ coincide.

When $J(\M{U})$ includes both $\Omega_{\M{W}}(\M{U})$ and $\Omega_{\Mtilde{W}}(\M{U}\Tra)$, rows and columns of $\M{U}$ are {\em simultaneously} shrunk towards each other as the parameter $\gamma$ increases. The penalized estimates exhibit the desired checkerboard structure as seen in \Fig{lung500_cba}. Note this shrinkage procedure is fundamentally different from methods like the clustered dendrogram, which independently cluster the rows and columns. By coupling row and column clustering, our formulation explicitly seeks out a solution with a checkerboard mean structure.

We now address choosing the weights $w_{ij}$ and $\tilde{w}_{ij}$. A judicious choice of the weights enables us to (i) employ a single regularization parameter $\gamma$, (ii) obtain more parsimonious clusterings, and (iii) speed up the convergence of key subroutines employed by COBRA. With these goals in mind, we recommend weights having following properties:
\begin{inparaenum}[(i)]
\item The column weights should sum to $1/\sqrt{n}$ and the row weights should sum to $1/\sqrt{p}$.
\item The column weight $w_{ij}$ should be inversely proportional to the distance between the $i$th and $j$th columns $\lVert \M{X}_{\cdot i} - \M{X}_{\cdot j} \rVert_2$. The row weights should be assigned analogously.
\item The weights should be sparse, namely consist mostly of zeros.
\end{inparaenum}

We now discuss the rationale behind our weight recommendations. The key to ensuring that a single tuning parameter suffices for identifying the checkerboard pattern, is keeping the two penalty terms
 $\Omega_{\M{W}}(\M{U})$ and $\Omega_{\Mtilde{W}}(\M{U}\Tra)$ on the same scale. If this does not hold, then either column or row clusterings will dominate. Consequently, since the columns are in $\Real^p$ and the rows are in $\Real^n$, we choose the column weights to sum to $1/\sqrt{p}$ and the row weights to sum to $1/\sqrt{n}$. Using weights that are inversely proportional to the distances between points, more aggressively shrinks rows and columns that are more similar to each other, and less aggressively shrinks rows and columns that are less similar to each other. Finally, sparse weights expedite computation. COBRA solves a sequence of convex clustering problems. The algorithm we employ for solving the convex clustering subproblem scales in storage and computational operations as $\mathcal{O}(npq)$, where $q$ is the number of non-zero weights. Shrinking all pairs of columns (rows) and taking $q = n^2$ ($q = p^2$) not only increases computational costs but also typically produces inferior clustering to sparser ones as seen in \cite{ChiLan2015}. We employ the sparse Gaussian kernel weights described in \cite{ChiLan2015}, which satisfy the properties outlined above. Additional discussion on this choice of weights is given in Web Appendix A.
  
\section{Properties of the Convex Formulation and COBRA's Solution}
\label{sec:solution}

The solution to minimizing \Eqn{biclust_objective_function} has several attractive properties as a function of the data $\M{X}$, the regularization parameter $\gamma$, and its weights $\M{W} = \{w_{ij}\}$ and $\Mtilde{W} = \{\tilde{w}_{ij}\}$, some of which can be exploited to expedite its numerical computation. We emphasize that these results are inherent to the minimization of the objective function \Eqn{biclust_objective_function}. They hold regardless of the algorithm used to find the minimum point, since they are a consequence of casting the biclustering problem as a convex program. Proofs of all propositions can be found in Web Appendix B. First, minimizing \Eqn{biclust_objective_function} is a well-posed optimization problem.
\begin{proposition}
\label{prop:existence_uniqueness} The function $F_\gamma(\M{U})$ defined in \Eqn{biclust_objective_function} has a unique global minimizer.
\end{proposition}
Furthermore, since $F_\gamma(\M{U})$ is convex, the only local minimum is the global minimum, and any algorithm that converges to a stationary point of $F_\gamma(\M{U})$ will converge to the global minimizer.
The next result will have consequences for numerical computation and is also the foundation underpinning COBRA's stability.
\begin{proposition}
\label{prop:solution_path_continuity} The minimizer $\M{U}^\star$ of \Eqn{biclust_objective_function} is jointly continuous in $(\M{X},\gamma, \M{W},\Mtilde{W})$.
\end{proposition}
Continuity of $\M{U}^\star$ in the regularization parameter $\gamma$ suggests employing warm-starts, or using the solution at one $\gamma$ as the starting point for a problem with a slightly larger $\gamma$, because small changes in $\gamma$ will result in small changes in $\M{U}^\star$. Continuity of $\M{U}^\star$ in $\M{X}$ tells us that the solution varies smoothly with perturbations in the data, our main stability result. Recall that the $i$th and $j$th columns (rows) are assigned to the same cluster if $\M{U}^\star_{\cdot i} = \M{U}^\star_{\cdot j}$ ($\M{U}^\star_{i\cdot} = \M{U}^\star_{j\cdot}$). Since $\M{U}^\star$ varies smoothly in the data, so do the differences $\M{U}^\star_{\cdot i} - \M{U}^\star_{\cdot j}$. While assigning entries to biclusters is an inherently discrete operation, we will see in \Sec{stability} that this continuity result manifests itself in assignments that are robust to perturbations in the data. Continuity of $\M{U}^\star$ in the weights also give us stability guarantees and assurances that COBRA's biclustering results should be locally insensitive to changes in the weights.

The parameter $\gamma$ tunes the complexity of COBRA's solution. We next show that COBRA's most complicated (small $\gamma$) and simplest (large $\gamma$) solutions coincide with the clustered dendrogram's most complicated and simplest solutions.
Clearly when $\gamma = 0$, the solution is just the data, namely $\M{U}^\star = \M{X}$. To get some intuition on how the solution behaves as $\gamma$ increases, observe that the penalty $J(\M{U})$ is a semi-norm. Moreover, under suitable conditions on the weights, spelled out in \As{connectedness} below,  $J(\M{U})$ is zero if and only if $\M{U}$ is a constant matrix. 
\begin{assumption}
\label{as:connectedness}
  For any pair of columns (rows), indexed by $i$ and $j$ with $i < j$, there exists a sequence of indices $i \rightarrow k \rightarrow \cdots \rightarrow l \rightarrow j$ along which the weights, $w_{ik}, \ldots, w_{lj}$ $(\tilde{w}_{ik}, \ldots, \tilde{w}_{lj})$ are positive.
\end{assumption}
\begin{proposition}
\label{prop:zero}
Under \As{connectedness}, $J(\M{U}) = 0$ if and only if $\M{U} = c\V{1}\V{1}\Tra$ for some $c \in \Real$.
\end{proposition}
This result suggests that as $\gamma$ increases the solution to the biclustering problem converges to the solution of the following constrained optimization problem:
\begin{eqnarray*}
\underset{\M{U}}{\min}\; \frac{1}{2} \lVert \M{X} - \M{U} \rVert_{\text{F}}^2 \quad \text{subject to $\M{U} = c\V{1}\V{1}\Tra$ for some $c \in \Real$},
\end{eqnarray*}
the solution to which is just the global mean $\Mbar{X}$, whose entries are all identically the average value of $\M{X}$ over all its entries. The next result formalizes our intuition that the centroids eventually coalesce to $\Mbar{X}$ as $\gamma$ becomes sufficiently large.
\begin{proposition}
\label{prop:coalesce}
Under \As{connectedness}, $F_{\gamma}(\M{U})$ is minimized by the grand mean $\Mbar{\M{x}}$ for $\gamma$ sufficiently large.
\end{proposition}
Thus, as $\gamma$ increases from 0, the centroids matrix $\M{U}^\star$ traces a continuous solution path that starts from $np$ biclusters, consisting of $\ME{U}{ij} = \ME{x}{ij}$, to a single bicluster, where $\ME{U}{ij} = (1/np)\sum_{i'j'} \ME{X}{i'j'}$ for all $i,j$.

\section{Estimation of Biclusters with COBRA}
\label{sec:algorithm}

Having characterized our estimator of the checkerboard means as the minimizer to \Eqn{biclust_objective_function}, we now turn to the task of computing it. From here on, we fix the data $\M{X}$ and the weights $\M{W}$ and $\Mtilde{W}$ and consider the biclustering solution as
a function of the parameter $\gamma$, denoting the solution as $\M{U}_\gamma$. 
The penalty term $J(\M{U})$ in \Eqn{biclust_objective_function} makes minimization challenging since it is non-smooth and not separable over any block partitioning of $\M{U}$. Coordinate descent \citep{FriHasH2007,WuLan2008} is an effective solution when the non-smooth penalty term is separable over some block partitioning of the variables, which is unfortunately not the case for \Eqn{biclust_objective_function}. Another popular iterative method for minimizing non-smooth convex functions is the alternating direction method of multipliers (ADMM) \citep{BoyParChu2011}. While an ADMM algorithm is feasible, we take a more direct approach with the Dykstra-like proximal algorithm (DLPA) proposed by \cite{BauCom2008} because it yields a simple meta-algorithm that can take advantage of fast solvers for the convex clustering problem.

DLPA generalizes a classic algorithm for fitting restricted least squares regression problems \citep{Dyk1983} and solves minimization problems of the form
\begin{eqnarray}
\label{eq:dlpa}
\underset{\M{U}}{\min}\; \frac{1}{2} \lVert \M{X} - \M{U} \rVert_{\text{F}}^2 + f(\M{U}) + g(\M{U}),
\end{eqnarray}
where $f$ and $g$ are lower-semicontinuous, convex functions. The biclustering problem is clearly an instance of \Eqn{dlpa}. Setting $f = \gamma \Omega_{\M{W}}$ and $g = \gamma\Omega_{\Mtilde{W}}$ in \Eqn{dlpa} gives us the pseudocode for COBRA shown in \Alg{COBRA}. The operation $\prox_{\gamma \Omega_{\M{W}}}(\M{Z})$ is the proximal mapping of the function $\gamma \Omega_{\M{W}}$ and is defined to be 
\begin{eqnarray*}
\prox_{\gamma \Omega_{\M{W}}}(\M{Z}) & = & \underset{\M{V}}{\arg\min}\;\left[\frac{1}{2} \lVert \M{Z} - \M{v} \rVert_{\text{F}}^2 + \gamma \Omega_{\M{W}}(\M{V})\right].
\end{eqnarray*}
Each proximal mapping in \Alg{COBRA} corresponds to solving a convex clustering problem.

The COBRA is very intuitive. The matrices $\M{Y}_m\Tra$ and $\M{U}_m$ are estimates of the means matrix at the $m$th iteration. The matrices $\M{P}_m$ and $\M{Q}_m$ encode discrepancies between these two estimates.  We alternate between clustering the rows of the matrix $\M{U}_m + \M{P}_m$ and the columns of the matrix $\M{Y}_m + \M{Q}_m$.  The following result guarantees that $\M{Y}_m\Tra$ and $\M{U}_m$ converge to the desired solution.
\begin{proposition}
\label{prop:COBRA} The COBRA iterates $\M{U}_m$ and $\M{Y}_m\Tra$ in \Alg{COBRA} converge to the unique global minimizer of the convex biclustering objective \Eqn{biclust_objective_function}.
\end{proposition}
\Prop{COBRA} not only ensures the algorithmic convergence of \Alg{COBRA}, but it also provides a natural stopping rule. We stop iterating once
$\lVert \M{U}_m - \M{Y}_m\Tra \rVert_{\text{F}}$ falls below some tolerance $\tau > 0$.
A proof of \Prop{COBRA}, as well as additional technical details and discussion on DLPA and COBRA, can be found in Web Appendix C.

\begin{algorithm}[t]
Set $\M{u}_0 = \M{x}, \M{p}_0 = \M{0}, \M{q}_{0} = \V{0}$ for $m=0, 1, \ldots$
\begin{algorithmic}[0]
  \caption{Convex BiclusteRing Algorithm (COBRA)}
  \label{alg:COBRA}
\Repeat
\State $\M{y}_{m} = \prox_{\gamma\Omega_{\Mtilde{W}}}(\M{u}_m\Tra + \M{p}_m\Tra)$
\Comment Convex Clustering of Rows
\State $\M{p}_{m+1} = \M{u}_m + \M{p}_m - \M{y}_m\Tra$
\State $\M{u}_{m+1} = \prox_{\gamma\Omega_{\M{W}}}(\M{y}_m\Tra + \M{q}_m\Tra)$
\Comment Convex Clustering of Columns
\State $\M{q}_{m+1} = \M{y}_m + \M{q}_m - \M{u}_{m+1}\Tra$
\Until{convergence}
\end{algorithmic}
\end{algorithm}

The advantage of using DLPA is that COBRA is agnostic to the actual algorithm used to solve the proximal mapping. This is advantageous since we cannot analytically compute $\prox_{\gamma\Omega_{\M{W}}}(\M{Z})$. In this paper we use the alternating minimization algorithm (AMA) introduced in \cite{ChiLan2015} to solve the convex clustering problem. The algorithm performs projected gradient ascent on the Lagrangian dual problem. Its main advantage is that it requires computational work and storage that is linear in the size of the data matrix $\M{X}$, when we use the sparse Gaussian kernel weights described in Web Appendix A. A second advantage is that hard clustering assignments are trivially obtained from variables employed in the splitting method. Nonetheless, the DLPA framework makes it trivial to swap in more efficient solvers that may become available in the future.

\section{Model Selection}
\label{sec:tuning_parameter}

Estimating the number of clusters or biclusters in a data set is a major challenge. With many existing biclustering methods, this is further exacerbated by the many tuning parameters that must be selected and the fact that biclustering assignments do not always change smoothly with the number of biclusters. For example, the sparse SVD method \citep{LeeSheHua2010} requires three tuning parameters: two parameters controlling the sparsity of the left and right singular vectors of the sparse SVD and one controlling its rank. Furthermore, selecting the number of biclusters and other tuning parameters can be a major computational burden for large data sets. For example, the sparse biclustering method \citep{TanWit2013} uses cross-validation to select the number of row and column partitions. This can be time consuming if a large range of possible number of row and column partitions are explored. In contrast, COBRA has one parameter $\gamma$, that controls both the number of biclusters and bicluster assignments; moreover, the number of biclusters and assignments varies smoothly with $\gamma$.

\subsection{Hold-Out Validation}
We present a simple but principled approach to selecting $\gamma$ in a data-driven manner by posing the model selection problem as another convex program.
We randomly select a hold-out set of elements in the data matrix and assess the quality of a model $\M{U}_\gamma$ on how well it predicts the hold-out set. This idea was first proposed by \cite{Wol1978} for model selection in principal component analysis and has been used more recently to select tuning parameters for matrix completion problems \citep{MazHasTib2010}. Denote these index pairs $\Theta \subset \{1, \ldots, p\} \times \{1, \ldots, n\}$, and let $\lvert \Theta \rvert$ denote the cardinality of the set $\Theta$. We may select a relatively small fraction of the elements, say 10\%, for validation, namely $\lvert \Theta \rvert \approx 0.1 \times np$. Denote the projection operator onto the set of indices $\Theta$ by $\mathcal{P}_{\Theta}(\M{X})$. The $ij$th entry of $\mathcal{P}_{\Theta}(\M{X})$ is $\ME{x}{ij}$ if $(i,j) \in \Theta$ and is zero otherwise.
We then solve the following convex optimization problem
\begin{eqnarray}
\label{eq:validation}
\underset{\M{U}}{\min}\; \tilde{F}_\gamma(\M{U}) & := &
\frac{1}{2} \lVert \mathcal{P}_{\Theta^c}(\M{x}) - \mathcal{P}_{\Theta^c}(\M{U}) \rVert_{\text{F}}^2 + \gamma J(\M{U})
\end{eqnarray}
for a sequence of $\gamma \in \mathcal{G} = \{\gamma_1=0, \ldots, \gamma_{\max}\}$. Recall that we denote the minimizer of $\tilde{F}_\gamma(\M{U})$ by $\M{U}_{\gamma}$.
We choose the $\gamma$ that minimizes the prediction error over the hold-out set $\Theta$, namely
$\gamma^\star = \underset{\gamma \in \mathcal{G}}{\arg\min}\; \lVert \mathcal{P}_\Theta(\M{X}) - \mathcal{P}_\Theta(\M{U}_\gamma) \rVert_\text{F}.$

\subsection{Solving the Hold-Out Problem}
\label{sec:MM_algorithm}

The problem defined in (\ref{eq:validation}) can be seen as a convex matrix completion problem. \Alg{MM} summarizes a simple procedure for reducing the problem of minimizing $\tilde{F}_\gamma(\M{U})$ to solving a sequence of the complete biclustering problems \Eqn{biclust_objective_function}. The solution from the previous iteration is used to fill in the missing entries in the current iteration; COBRA is then applied to the complete data matrix. This approach is identical to the soft-impute approach of \cite{MazHasTib2010} for solving the matrix completion problem using a nuclear norm penalty instead of our fusion penalty $J(\M{U})$. The similarity is not a coincidence as both procedures are instances of a majorization-minimization (MM) algorithm \citep{LanHunYan2000} which apply the same majorization on the smooth quadratic term. We defer details on this connection to Web Appendix D. \Alg{MM} has the following convergence guarantees for the imputation algorithm.
\begin{proposition}
\label{prop:mm_algorithm}
The limit points of the sequence of iterates $\Mn{U}{m}$ of \Alg{MM} are solutions to \Eqn{validation}.
\end{proposition}
Thus, we have turned the model selection problem of selecting both the number of biclusters and bicluster assignments into a principled convex program with strong convergence guarantees. 

\begin{algorithm}[t]
\begin{algorithmic}[1]
  \caption{COBRA with missing data}
  \label{alg:MM}
\State Initialize $\Mn{U}{0}$.
\Repeat
\State $\M{M} \gets \mathcal{P}_{\Theta^c}(\M{X}) + \mathcal{P}_{\Theta}(\Mn{U}{m})$
\State $\Mn{U}{m+1} \gets \text{COBRA}(\M{M})$
\Until{convergence}
\end{algorithmic}
\end{algorithm}

\section{Simulation Studies}
\label{sec:comparison}
We compare COBRA and two other biclustering methods that also assume an underlying checkerboard mean structure. The first is the clustered dendrogram; hard biclustering assignments are made for the clustered dendrogram using the widely used dynamic tree cutting algorithm \citep{LanZha2008} implemented in the package {\tt dynamicTreeCut}. The second is the sparse biclustering method \citep{TanWit2013} implemented in the package {\tt sparseBC} All parameters were selected according to methods in the R packages.

Assessing the quality of a clustering is almost as hard as the clustering problem itself, as evidenced by a plethora of quantitative measures for comparing how similar two clusterings are. In this paper, we use the following three measures: the Rand index (RI), the adjusted Rand index (ARI), and the variation of information (VI).  We included the RI \citep{Ran1971}, since it is one of the most widely used criteria for comparing partitions; it maps a pair of partitions to a number between 0 and 1, where 1 indicates perfect agreement between two partitions.  Despite its popularity, the RI has some limitations (See Web Appendix E). Consequently, we also use the ARI \citep{HubAra1985}, which was engineered to address deficiencies in the RI. Like the RI, the ARI takes a maximum value of 1 if the clustering solution is identical to the true structure and takes a value close to 0 if the clustering result is obtained from random partitioning. Finally, we compared clustering results using the VI \citep{Meila2007}. Unlike the RI and ARI, the VI is a metric. Consequently, under the VI we can speak rigorously about a neighborhood around a given clustering. In fact, the nearest neighbors of a clustering $\mathcal{C}$ under the VI metric are clusterings obtained by splitting or merging small clusters in $\mathcal{C}$. While not as popular as the RI and ARI, the VI is perhaps the most appropriate for assessing two hierarchical clusterings. As a metric, the VI takes a minimum value of 0 when there is perfect agreement between two partitions. Definitions of these criteria are given in Web Appendix E.

To simulate data with a checkerboard partition mean pattern, we consider the partition of the data matrix as the set of biclusters induced by taking the cross-product of the row and column groups. For all experiments, we used a single fixed weight assignment rule (See Web Appendix A), therefore COBRA selects a single tuning parameter $\gamma$ by validation. We perform computations in serial on a multi-core computer with 24 3.3 GHz Intel Xeon processors and 189 GB of RAM.  Note that run times may vary depending on input parameters chosen. For example, COBRA estimates can be computed to high accuracy at greater computation, and sparse biclustering can explore a wider range of candidate row and column clusters at greater computation. To be fair, we did not cherry pick these parameters but picked some reasonable values and used them throughout our numerical experiments. For example, in sparse biclustering when there are 8 true row clusters and 8 true column clusters, we set the range of row and column clusters to be 1 to 12.

Finally, we also considered two variants of COBRA that employ standard refinements on the Lasso: the adaptive Lasso \citep{Zou2006} and the thresholded Lasso \citep{Meinshausen2009}. These refinements address the well known issue that the Lasso tends to select too many variables. Thus, we anticipate that COBRA estimates may identify too many biclusters. Consequently, we also compared the performance of an adaptive COBRA and thresholded COBRA in our study. Details on these refinements are in Web Appendix G.

We compare COBRA, the clustered dendrogram, and the sparse biclustering (spBC) algorithm of \cite{TanWit2013} on their abilities to recover a checkerboard pattern. 
Again for the clustered dendrogram, row and column assignments are made with the dynamic tree cutting (DCT) method \citep{LanZha2008}. We simulate a $200\times200$ data matrix with a checkerboard bicluster structure, where $\ME{X}{ij} \sim$ iid~$N(\mu_{rc},\sigma^2)$ and $\mu_{rc}$ took on one of 25 equally spaced values between -6 and 6, namely $\mu_{rc} \sim$ Uniform$\{-6,-5.5,\ldots,5.5,6\}$. We consider a high signal-to-noise ratio (SNR) situation where the minimum difference among bicluster means ($0.5$) is comparable to the noise ($\sigma = 1.5$) and a low SNR one where the minimum difference is dominated by the noise ($\sigma=3.0$). 
To assess the performance as the number of column and row clusters are varied, we generated data using 16, 32, and 64 biclusters, corresponding to 2, 4, and 8 row groups and 8 column groups, respectively.
Since typical clusters will not be equal in size, rows and columns are assigned to each of the groups randomly according to a non-uniform distribution. The probability that a row is assigned to the $i$th group is inversely proportional to $i$. Columns are assigned analogously to groups. 

We computed the RI, ARI, and VI between the true biclusters and those obtained by the COBRA variants, DCT, and spBC. \Tab{gaussian_checkerboard} reports the average RI, ARI, and VI over 50 replicates as well as the average number of recovered biclusters $\hat{N}_b$ and run times; $N_b$ denotes the true number of biclusters. In the high SNR scenario, all methods do well across all measures. In the low SNR scenario, the COBRA variants often perform nearly as well or better than the other two methods. While DCT is significantly faster than all other methods, it also performs the worst at recovering the true biclusters in the low SNR scenario. The run times for COBRA indicate that it is computationally competitive compared to alternative biclustering solutions.

\begin{table}[th]
\begin{tabular}{ l  c c c c c c c }
& $N_b$ & $\sigma$ & COBRA & COBRA (A) & COBRA (T) & DCT & spBC \\ \hline
RI & 16 & 1.5 & 0.993 & {\bf 0.994} & 0.993 & 0.952 &  0.969 \\  
& 32 & 1.5 &  {\bf 0.999}  & {\bf 0.999}  &  {\bf 0.999}  & 0.993  & 0.997 \\
& 64 & 1.5 &  {\bf 0.999}  & {\bf 0.999}  &  {\bf 0.999}  &  {\bf 0.999}  &  {\bf 0.999} \\
& 16 & 3.0 & 0.971  &  {\bf 0.982}  & 0.959  & 0.944  & 0.966 \\  
& 32 & 3.0 &  {\bf 0.997}  &  {\bf 0.997}  & 0.996  & 0.990  &  {\bf 0.997} \\  
& 64 & 3.0 &  {\bf 0.999}  &  {\bf 0.999}  &  {\bf 0.999}  &  {\bf 0.999}  &  {\bf 0.999} \\ \hline
ARI & 16 & 1.5 &  {\bf 0.952} & 0.924 & 0.950  & 0.713 & 0.804 \\ 
& 32 & 1.5 & 0.995  & 0.978  &  {\bf 0.996}  & 0.916  & 0.965 \\
& 64 & 1.5 &  {\bf 0.999}  & 0.996  &  {\bf 0.999}  & 0.981  & 0.995 \\
& 16 & 3.0 & 0.798  &  {\bf 0.909}  & 0.741  & 0.449  & 0.784 \\
& 32 & 3.0 & 0.958  &  {\bf 0.982}  & 0.945  & 0.844  & 0.958 \\
& 64 & 3.0 & 0.992  & 0.992  & 0.985  & 0.979  &  {\bf 0.993} \\ \hline
VI & 16 & 1.5 &  {\bf 0.117}  & 0.132  &  {\bf 0.117}  & 0.713 & 0.180 \\
& 32 & 1.5 & 0.013  & 0.032  &  {\bf 0.009}  & 0.195  & 0.150 \\
& 64 & 1.5 &  {\bf 0.001}  & 0.022  &  {\bf 0.001}  & 0.043  & 0.271 \\
& 16 & 3.0 & 0.627  & 0.446  & 0.730  & 2.245  &  {\bf 0.260} \\
& 32 & 3.0 &  {\bf 0.097}  & 0.125  & 0.127  & 0.516  & 0.166 \\
& 64 & 3.0 &  {\bf 0.020}  & 0.046  & 0.034  & 0.061  & 0.181 \\ \hline
$\hat{N}_{b}$ & 16 & 1.5 & 18.8  & 25.3 & {\bf 16.8}  & 10.7  & 14.7 \\
& 32 & 1.5 & 32.9  & 38.1  & {\bf 32.1}  & 28.8  & 31.2 \\
& 64 & 1.5 & {\bf 64.3}  & 69.7  & {\bf 64.3}  & 62.5  & 60.3 \\
& 16 & 3.0 & 43.7  & 46.3  & 30.3  & 54.1  & {\bf 13.9} \\
& 32 & 3.0 & {\bf 30.4}  & 44.2  & 29.9  & 44.7  & 30.1 \\
& 64 & 3.0 & 65.1  & 77.1  & {\bf 63.2}  & 65.9  & 61.0 \\ \hline
time (sec) & 16 & 1.5 & 26.83  & 50.36 & 27.30 & {\bf 0.21}  & 558.27 \\
& 32 & 1.5 & 30.07  & 57.49  & 30.49  & {\bf 0.21}  & 401.96 \\
& 64 & 1.5 & 29.93  & 58.54  & 30.35  & {\bf 0.21}  & 288.91 \\
& 16 & 3.0 & 35.55  & 67.56  & 36.02  & {\bf 0.22}  & 564.85 \\
& 32 & 3.0 & 33.99  & 66.06  & 34.44  & {\bf 0.21}  & 432.71 \\
& 64 & 3.0 & 33.09  & 65.10  & 33.50  & {\bf 0.20}  & 284.24 \\ \hline \\
\end{tabular}
\caption{\label{tab:gaussian_checkerboard} 
Checkerboard mean structure with iid $N(0,\sigma^2)$ noise: low-noise ($\sigma = 1.5$) and high-noise ($\sigma = 3.0$). COBRA (A) is the adaptive COBRA and COBRA (T) is the thresholded COBRA. Details on these two variants are in Web Appendix G.}
\end{table}
In closing our discussion on simulations, we reiterate that COBRA is designed to recover checkerboard patterns. While checkerboard patterns feature prominently in a range of applications, we also acknowledge that they are not universal. Nonetheless, by examining both the estimated biclusters and bicluster means, COBRA can potentially identify the correct biclusters even when the checkerboard assumption is violated. We discuss how COBRA can accomplish this in more detail with a case study in Web Appendix F. 

\section{Application to Genomics}
\label{sec:stability}

To illustrate COBRA in action on a real example, we revisit the lung cancer data studied by \cite{LeeSheHua2010}.  We have selected the 500 genes with the greatest variance from the original collection of 12,625 genes.
Subjects belong to one of four subgroups; they are either normal subjects (Normal) or have been diagnosed with one of three types of cancers: pulmonary carcinoid tumors (Carcinoid), colon metastases (Colon), and small cell carcinoma (Small Cell). 

We first illustrate how the solution $\M{U}_\gamma$ evolves as $\gamma$ varies. \Fig{lung_cba_path} shows snap shots of the COBRA solution path of this data set, as the parameter $\gamma$ increases. The path captures the whole range of behavior between under-smoothed estimates of the mean structure (small $\gamma$), where each cell is assigned its own bicluster, to over-smoothed estimates (large $\gamma$), where all cells belong to a single bicluster.  In between these extremes, we see rows and columns ``fusing" together as $\gamma$ increases. Thus we have visual confirmation that minimizing \Eqn{biclust_objective_function} over a range of $\gamma$, yields a convex formulation of the clustered dendrogram.

While generating the entire solution path enables us to visualize the hierarchical relationships between biclusterings for different $\gamma$, we may ultimately require a hard biclustering assignment. By applying the validation procedure described in \Sec{tuning_parameter}, we arrive at the smoothed mean estimate shown previously in \Fig{lung500_cba}

We next conduct a simulation experiment based on the lung cancer data to test the stability and reproducibility of biclustering methods, critical qualities for real scientific analysis. To expedite computation in these experiments, we restrict our attention to the 150 genes with the highest variance. We first apply the biclustering methods on the original data to obtain baseline biclusterings. We then add iid $N(0,\sigma^2)$, noise where $\sigma = 0.5, 1.0, 1.5$ to create a perturbed data set on which to apply the same set of methods. We compute the RI, ARI, and VI between the baseline clustering and the one obtained on the perturbed data.  \Tab{stability} shows the average RI, ARI, and VI of 50 replicates as well as run times. For all values of $\sigma$, we see that the two COBRA variants tend to produce the most stable and reproducible results. The fact that the ARI scores are poor for plain COBRA but the RI and VI scores are good indicate that COBRA tends to shrink the same sets of rows and columns together even if it fails to fuse them together consistently. Again the run time results indicate that COBRA is computationally competitive. For completeness, results from an identical stability study for methods that do not assume a checkerboard pattern can be found in Web Appendix H.

\begin{figure}
        \centering
        \subfloat[$\gamma = 0$]{\label{fig:lung0}%
        \includegraphics[width=1.7in]{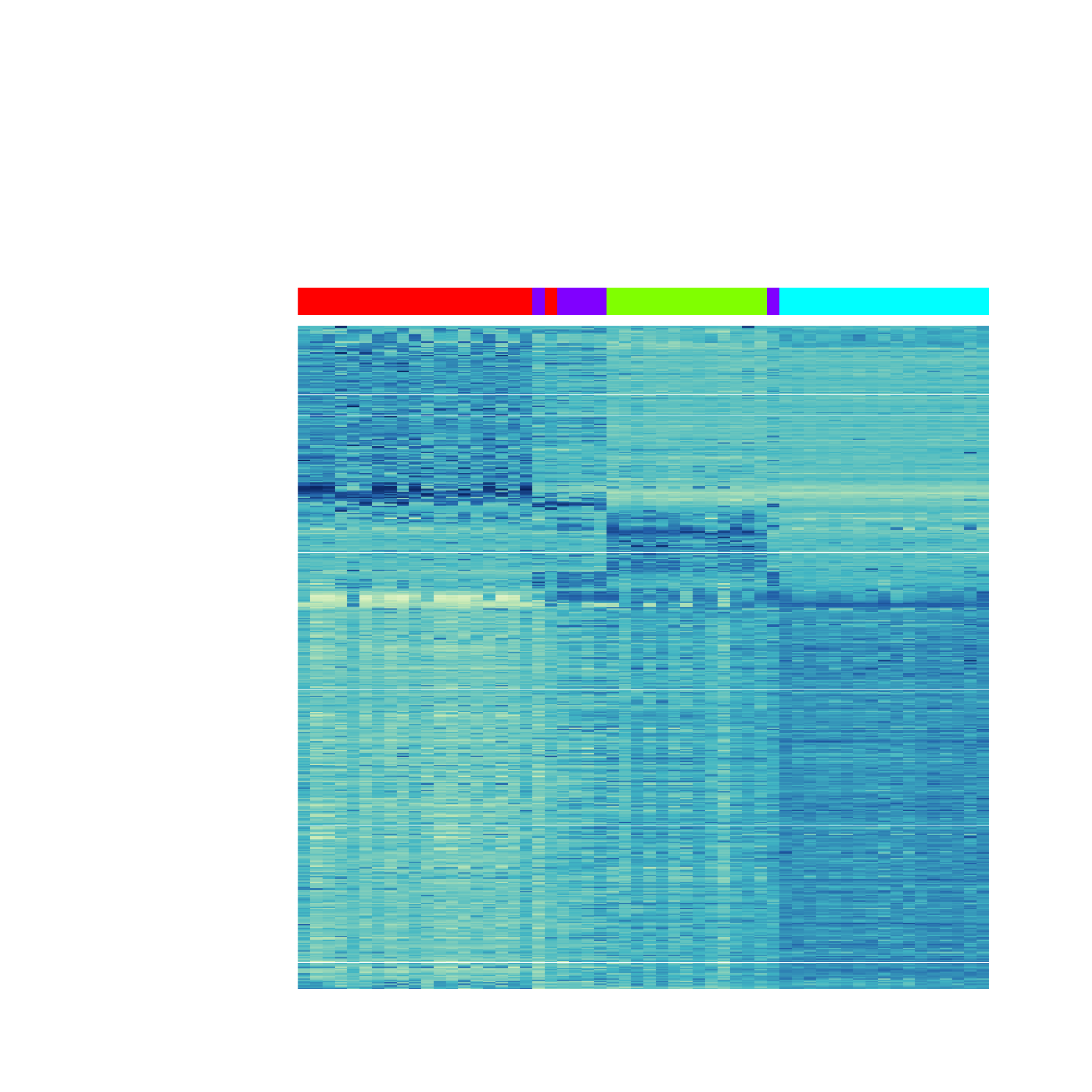}} \hspace{-1cm}
        \subfloat[$\gamma = 10^{1.45}$]{\label{fig:lung5}%
        \includegraphics[width=1.7in]{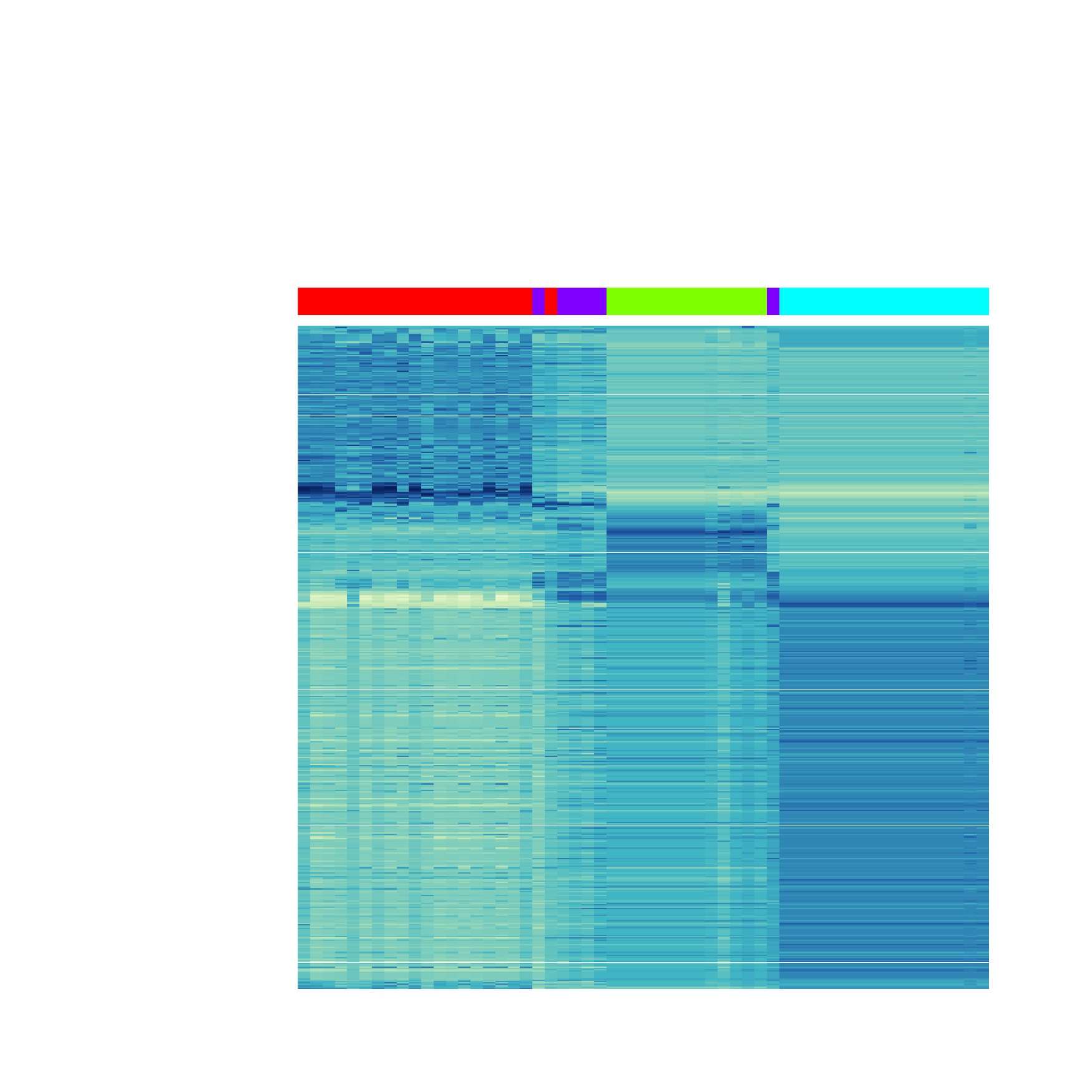}} \hspace{-1cm}  
        \subfloat[$\gamma = 10^{1.79}$]{\label{fig:lung7}%
        \includegraphics[width=1.7in]{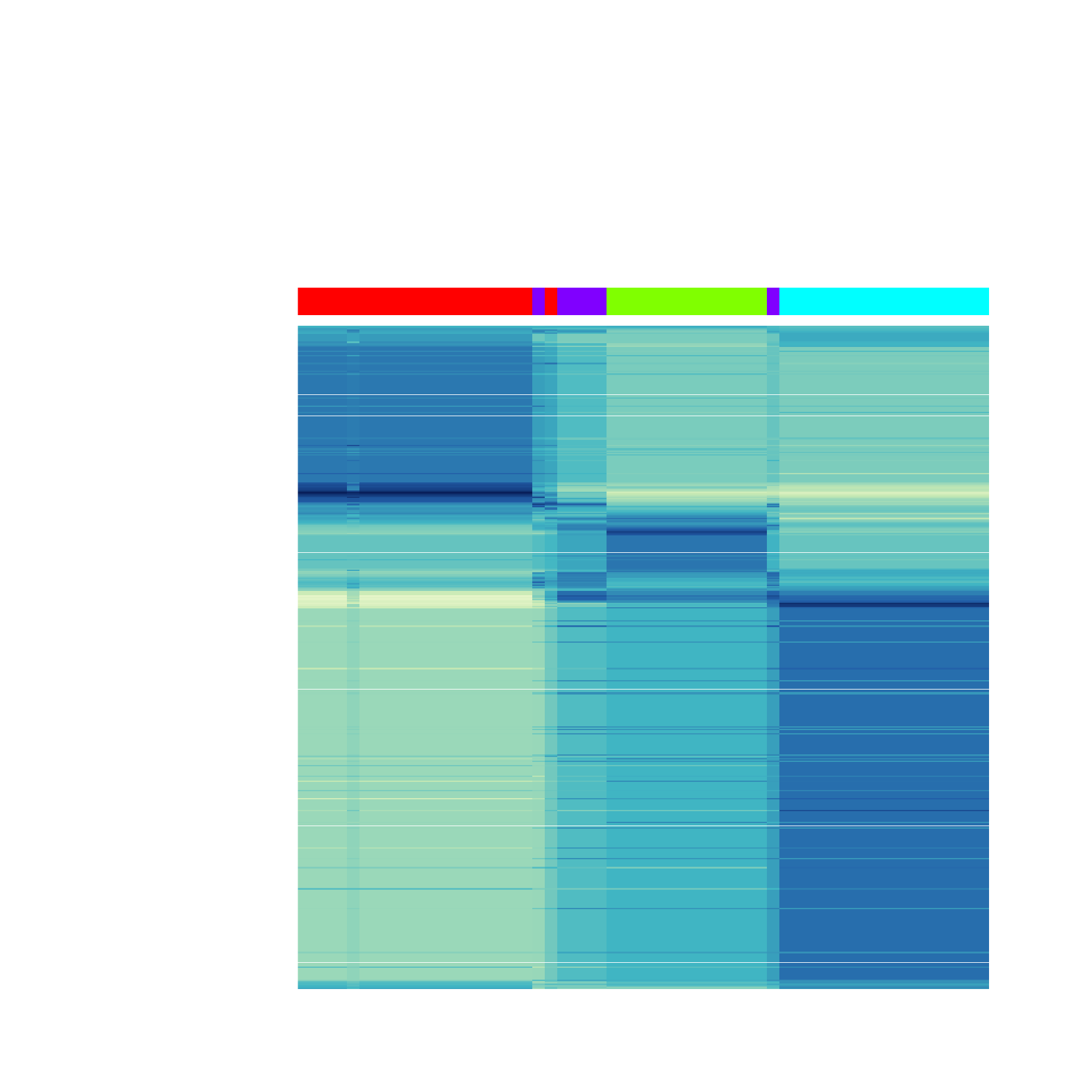}} \hspace{-1cm}
        \subfloat[$\gamma = 10^{2.01}$]{\label{fig:lung9}%
        \includegraphics[width=1.7in]{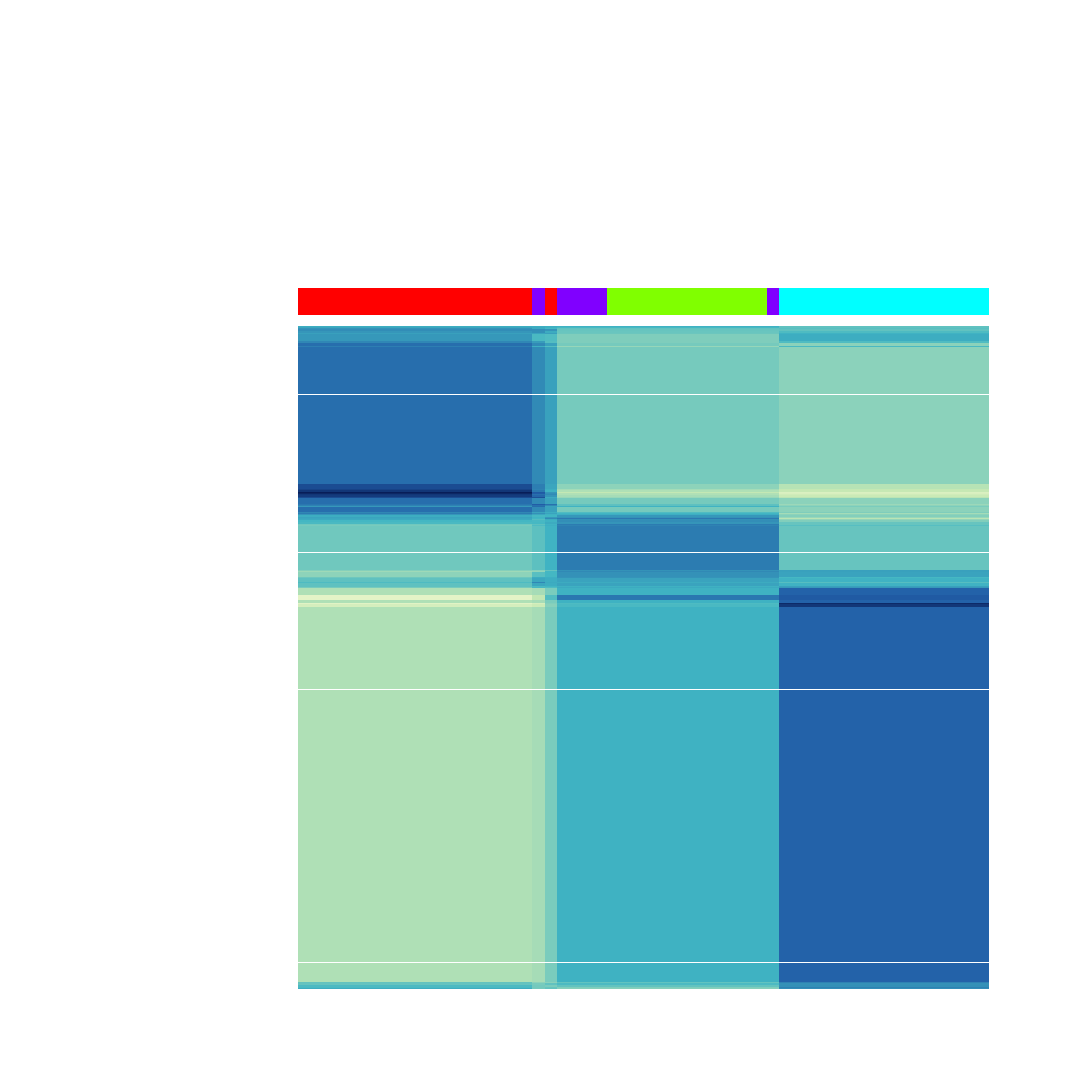}} \\ \vspace{-0.5cm}
        \subfloat[$\gamma = 10^{2.24}$]{\label{fig:lung12}%
        \includegraphics[width=1.7in]{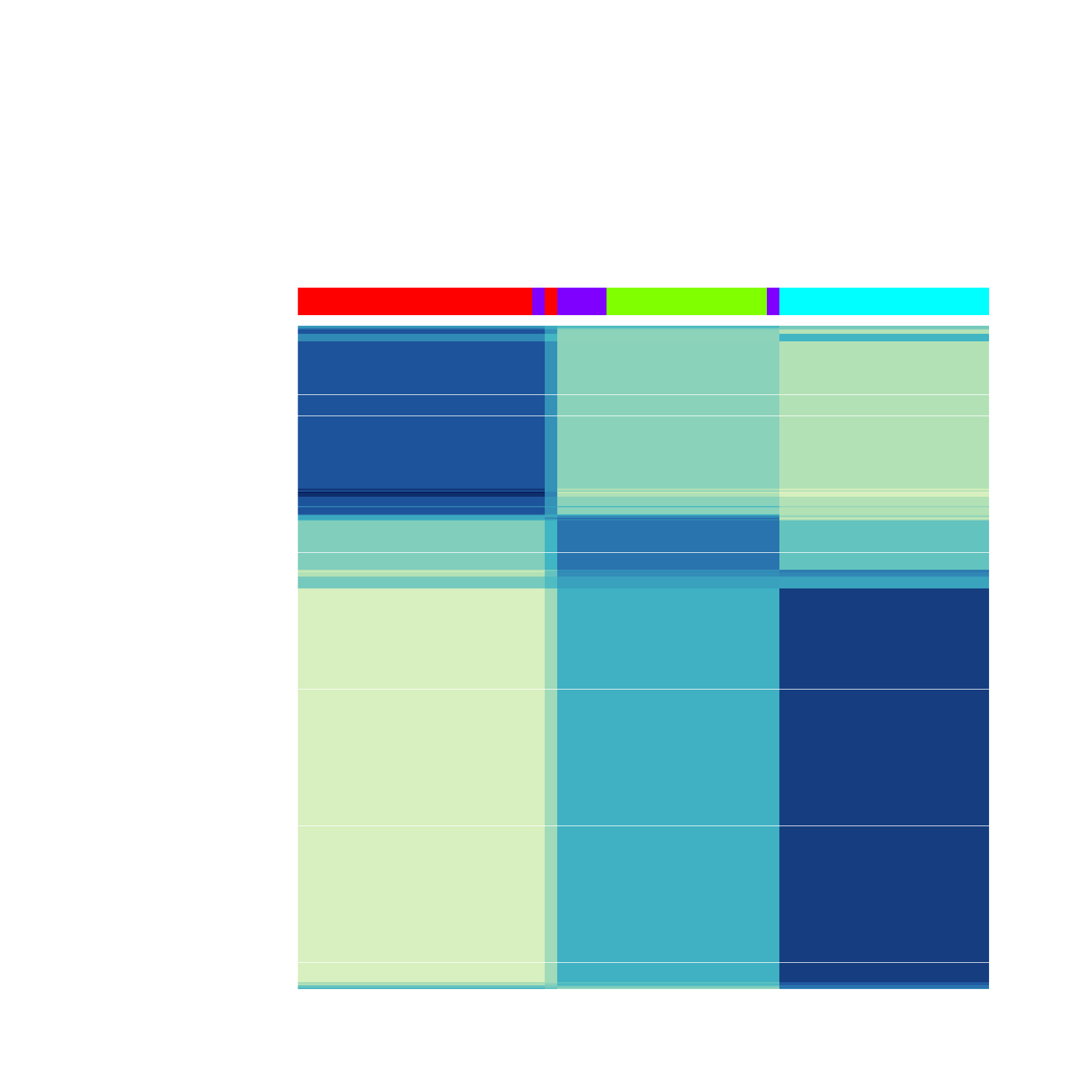}} \hspace{-1cm}
        \subfloat[$\gamma = 10^{2.35}$]{\label{fig:lung17}%
        \includegraphics[width=1.7in]{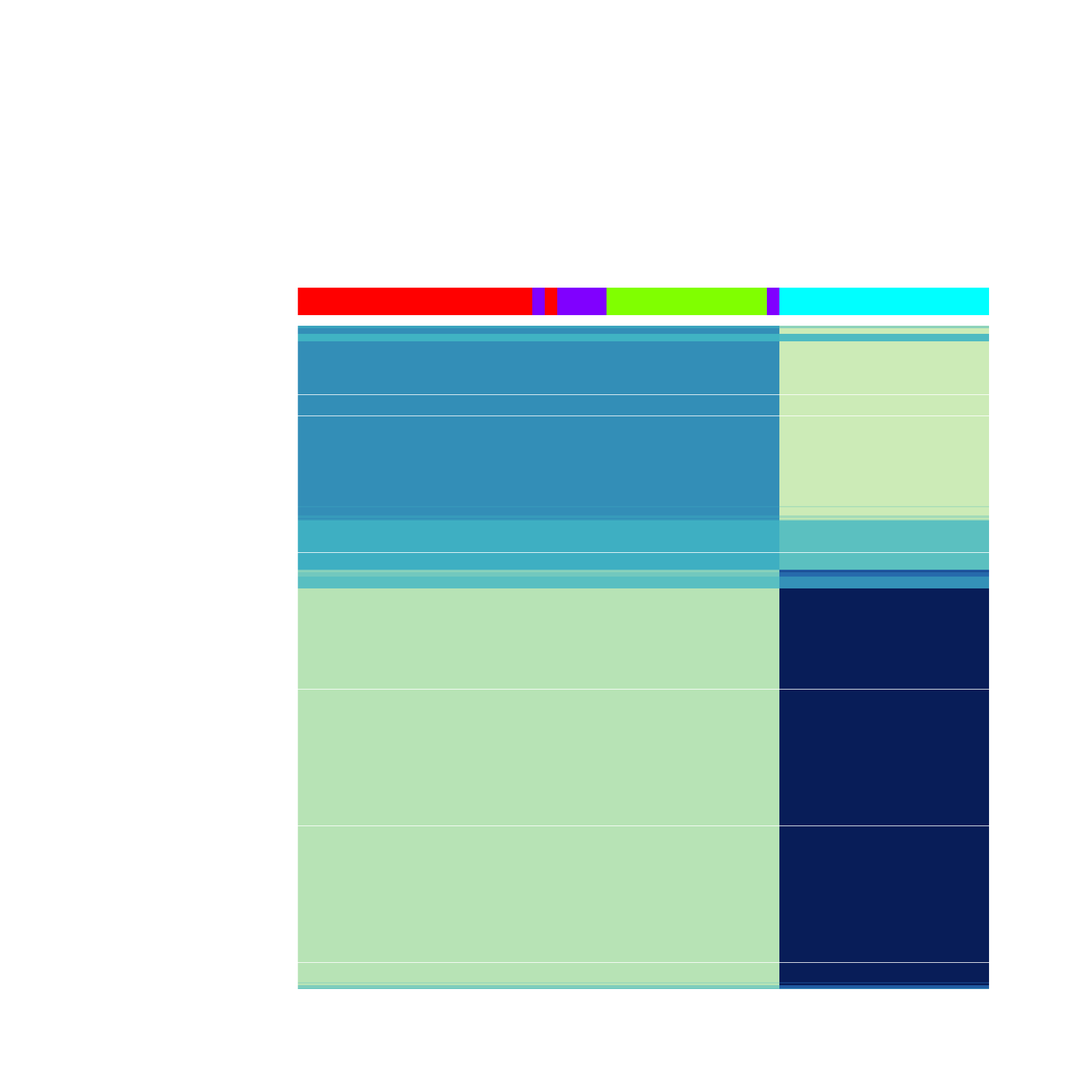}} \hspace{-1cm}
        \subfloat[$\gamma = 10^{3.03}$]{\label{fig:lung18}%
        \includegraphics[width=1.7in]{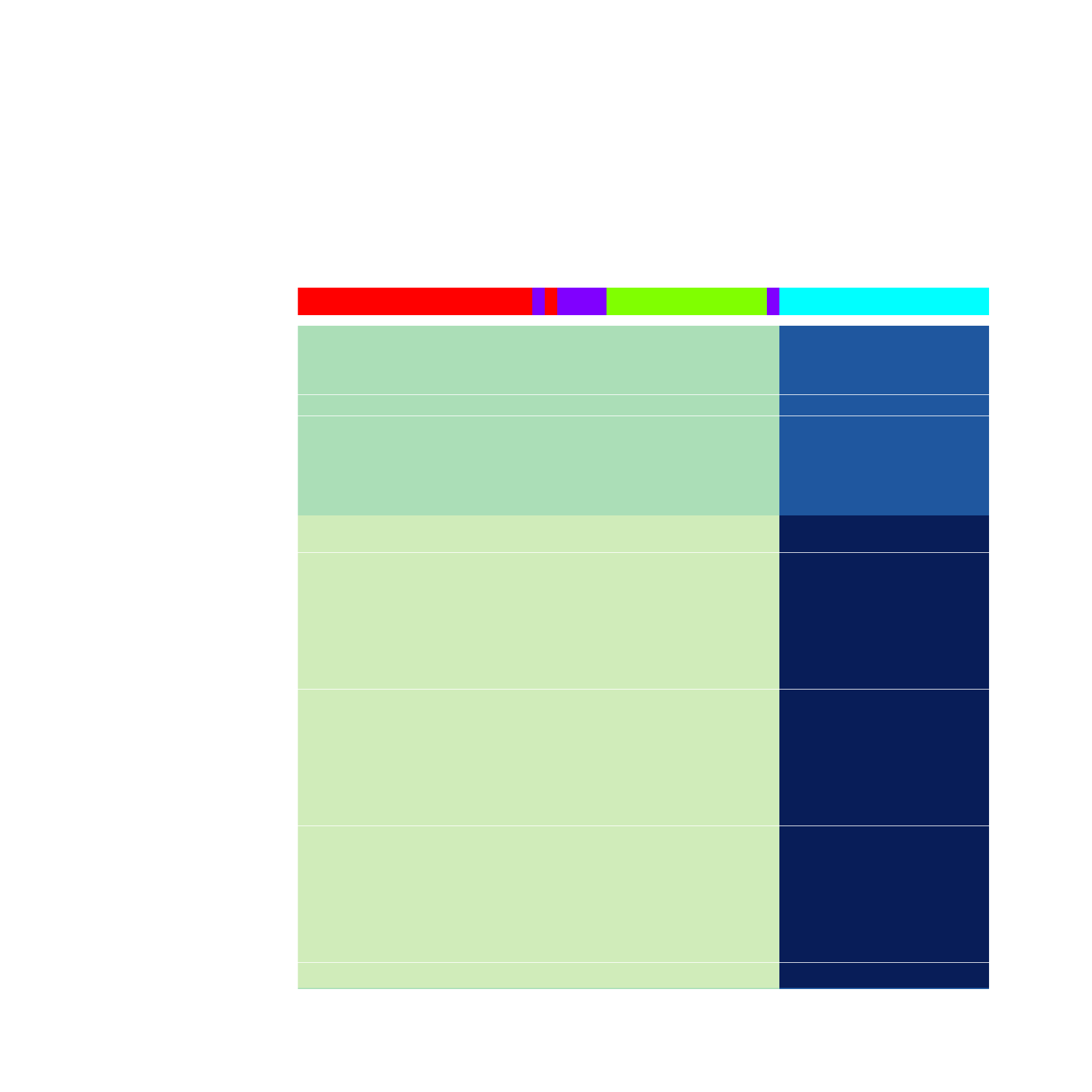}} \hspace{-1cm}
        \subfloat[$\gamma = 10^{3.14}$]{\label{fig:lung20}%
        \includegraphics[width=1.7in]{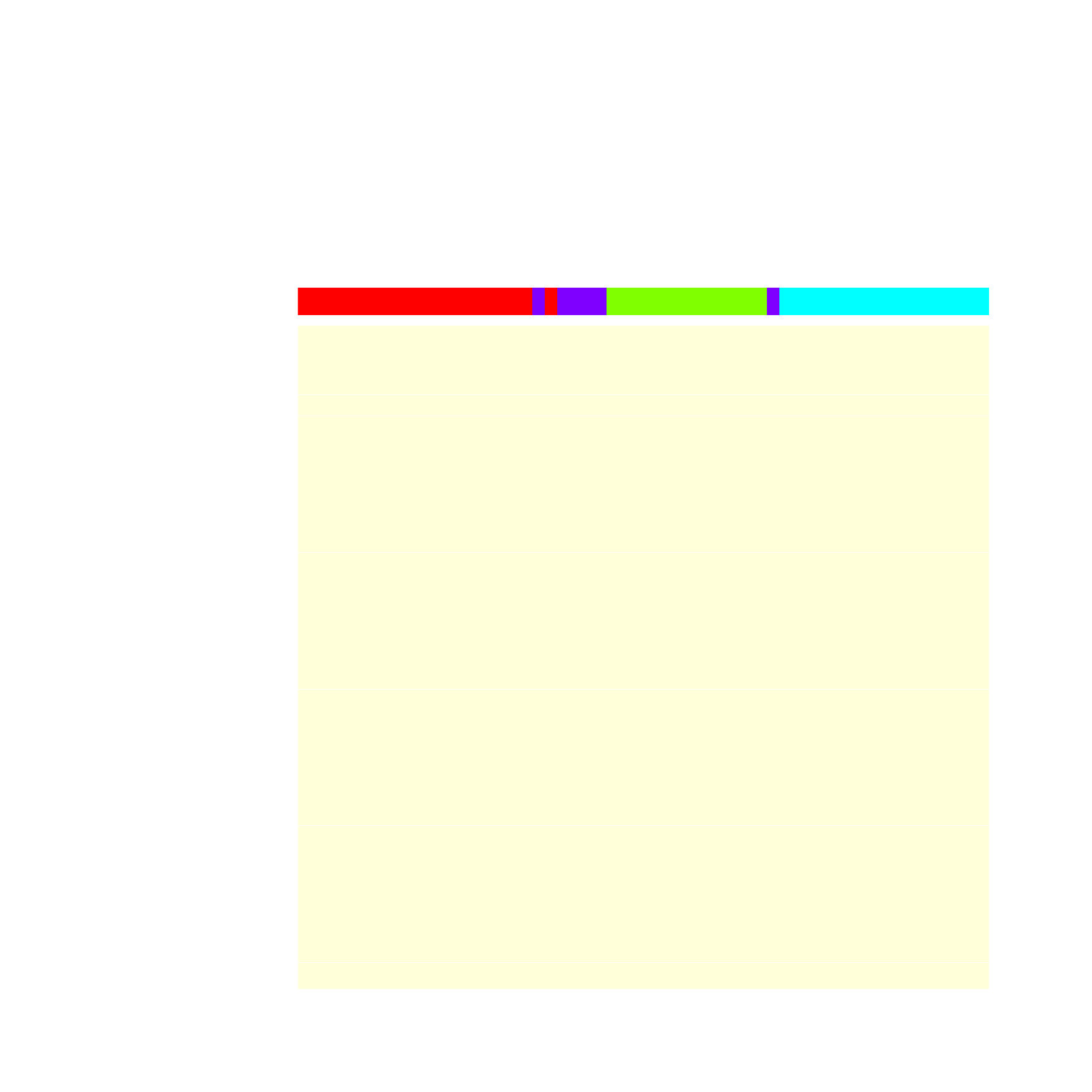}} \\
     \caption{Snap shots of the COBRA solution path of the lung cancer data set, as the parameter $\gamma$ increases. The path captures the whole range of behavior between under-smoothed estimates of the mean structure (small $\gamma$), where each cell is assigned its own bicluster, to over-smoothed estimates (large $\gamma$), where all cells belong to a single bicluster.}
     \label{fig:lung_cba_path}
\end{figure}

\begin{table}[th]
\begin{tabular}{ l c c c c c c }
& $\sigma$ & COBRA & COBRA (A) & COBRA (T) & DCT & spBC \\ \hline
RI & 0.5 & 0.984 & {\bf 0.992} & 0.959 & 0.979 & 0.974 \\ 
& 1.0 & 0.981 & {\bf 0.990} & 0.944 & 0.974 & 0.965 \\
& 1.5 & 0.973 & {\bf 0.989} & 0.896 & 0.973 & 0.936 \\ \hline
ARI & 0.5 & 0.350 & 0.788 & {\bf 0.813} & 0.530 & 0.642 \\ 
& 1.0 & 0.233 & 0.686 & {\bf 0.766} & 0.439 & 0.544 \\
& 1.5 & 0.201 & {\bf 0.667} & 0.644 & 0.340 & 0.397 \\ \hline
VI & 0.5 & 1.924 & 0.882 & {\bf 0.776} & 2.120 & 1.568 \\ 
& 1.0 & 2.380 & 1.276 & {\bf 0.962} & 2.769 & 2.174 \\
& 1.5 & 2.721 & {\bf 1.312} & 1.320 & 3.505 & 2.915 \\ \hline
time (sec) & 0.5 & 15.44 & 23.46 & 15.65 & {\bf 0.07} & 151.59 \\ 
& 1.0 & 25.21 & 34.51 & 25.53 & {\bf 0.11} & 197.00 \\
& 1.5 & 18.18 & 26.43 & 18.50 & {\bf 0.12} & 207.88 \\ \hline \\
\end{tabular}
\caption{\label{tab:stability} 
Stability and reproducibility of biclusterings in lung cancer microarray data. COBRA variants, the clustered dendrogram with dynamic tree cutting, and sparse Biclustering are applied to the lung cancer data to obtain baseline biclusterings. We then perturb the data by adding iid\@ $N(0,\sigma^2)$ noise where $\sigma = 0.5$ (Small Pert.), 1.0 (Medium Pert.), 1.5 (Large Pert.).} 
\end{table}

\section{Discussion}
\label{sec:discussion}

Our proposed method for biclustering, COBRA, can be considered a principled reformulation of the clustered dendrogram. Unlike the clustered dendrogram, COBRA returns a unique global minimizer of a goodness-of-fit criterion, but like the clustered dendrogram, COBRA is simple to interpret. COBRA also sports two key improvements over existing biclustering methods. First, it is more stable. COBRA biclustering assignments on perturbations of the data agree noticeably more frequently than those of existing biclustering algorithms. Second, it admits an effective and efficient model selection procedure for selecting the number of biclusters, that reduces the problem to solving a sequence of convex biclustering problems. The upshot of these two qualities is that COBRA produces results that are both simple to interpret and reproducible.

The simplicity of our means model is also its greatest weakness, since we consider only checkerboard patterns, namely we 
assign each observation to exactly one bicluster and do not consider overlapping biclusters \citep{CheChu2000,LazOwe2002,ShaWeiNob2009}. Nonetheless, while models that allow for overlapping biclusters might be more flexible, they are also harder to interpret.

While our simulation studies demonstrated the effectiveness of COBRA, there is room for improvement. We highlight an intriguing suggestion made during the review of this article. In many real-world applications there is no ``true" fixed number of biclusters. Instead, the underlying latent structure may be a continuum of biclusters at different scales of row and column aggregation. Indeed, COBRA has the potential to estimate a multiscale model of the data. When the weights are uniform, all columns (rows) are averaged together. When the weights are positive only among nearest neighbors, only nearest neighboring columns (rows) are averaged together. Thus, by tuning the weights, we can obtain smoothed estimates of the data at different scales of resolution. 

The ability to smooth estimates at different scales suggests a connection to computational harmonic analysis. Indeed, \cite{CoiGav2011} explore the biclustering problem through a wavelet representation of the data matrix. They also seek a representation that is smooth with respect to partitions of the row and column graphs that specify the similarity among the observations and features. A checkerboard mean structure at different scales can be obtained via operations in the wavelet domain, namely by thresholding wavelet coefficients corresponding to different scales. We are currently exploring how to adapt \cite{CoiGav2011}'s strategy to solve a sequence of COBRA problems at different scales in order to recover a continuum of biclusters.

An R package, called cvxbiclustr, implementing COBRA is available on CRAN.


\section*{Acknowledgements}
EC acknowledges support from CIA Postdoctoral Fellowship 2012-12062800003. GA acknowledges support from NSF DMS 1209017, 1264058, and 1317602.
RB acknowledges support from ARO MURI W911NF-09-1-0383 and AFOSR grant FA9550-14-1-0088.
\appendix

\section*{Web Appendix A. Column and Row Weights}
\label{sec:weights}

\setcounter{equation}{3}
Recall that our goal is to minimize the following convex criterion
\begin{eqnarray}
\label{eq:biclust_objective_function}
F_{\gamma}(\M{U}) & = & \frac{1}{2} \lVert \M{x} - \M{U} \rVert_{\text{F}}^2 + \gamma \underbrace{\left [\Omega_{\M{w}}(\M{U}) + \Omega_{\Mtilde{W}}(\M{U}\Tra) \right ]}_{J(\M{U})},
\end{eqnarray}
where $\Omega_{\M{w}}(\M{U}) = \sum_{i<j}w_{ij} \|\M{U}_{\cdot i}-\M{U}_{\cdot j} \rVert_2$, and $\M{U}_{\cdot i}$ ($\M{U}_{i \cdot})$ denotes the $i$th column (row) of the matrix $\M{U}$. In this work, we use the sparse Gaussian kernel weights proposed in \cite{ChiLan2015} for the weights $\M{W}$ and $\Mtilde{W}$ that define
the terms $\Omega_{\M{W}}(\M{U})$ and $\Omega_{\Mtilde{W}}(\M{U})$.

We construct the weights in two steps. We describe these steps for computing the column weights; the row weights are computed analogously. We start by computing pre-weights between the $i$th and $j$th columns as $\hat{w}_{ij} = \iota^k_{\{i,j\}} \exp(-\phi \lVert \M{x}_{\cdot i} - \M{x}_{\cdot _j} \rVert_2^2)$, as the product of two terms. The first factor $\iota^k_{\{i,j\}}$ is 1 if $j$ is among $i$'s $k$-nearest-neighbors or vice versa and 0 otherwise. The first term controls the sparsity of the weights. The second factor is a Gaussian kernel that puts greater pressure on similar columns to fuse and less pressure on dissimilar columns to fuse. The nonnegative constant $\phi$ controls the rate at which the pressure to fuse is applied as a function of the distance between columns; the value $\phi = 0$ corresponds to uniform weights.  The pre-weights $\hat{w}_{ij}$ are then normalized to sum to $1/\sqrt{p}$.

For all experiments in the paper, we set $\phi = 0.5$ and set $k = 10$ for both row and column weights.

We briefly discuss the rationale behind our weight choice here and refer readers to \cite{ChiLan2015} for a more detailed exposition. 
\cite{ChiLan2015} give several examples that show that restricting positive weights to nearest neighbors enhances both computational efficiency and clustering quality. In their examples they showed that if dense Gaussian kernel weights were used, cluster centroids shrunk towards each other 
as the tuning parameter $\gamma$ increased but no fusions would occur along the path save a single simultaneous fusion of all cluster centroids for a sufficiently large $\gamma$. Thus, while the two factors defining the weights act similarly, sensible fusions along the solution path could be achieved only by using them together. This is best illustrated in the half-moons example in \citep{ChiLan2015}.


\section*{Web Appendix B. Proofs of Solution Properties}

In this appendix, we give proofs of propositions in Section 3 of our paper.

\setcounter{section}{3}

\begin{proposition}[Existence and Uniqueness]
\label{prop:existence_uniqueness} The function $F_\gamma(\M{U})$ defined in (1) has a unique global minimizer.
\end{proposition}
We first recall a few definitions and concepts useful in optimization \citep{Lan2013}. A function is {\em coercive} if all its sub level sets are compact. A function $f$ is {\em convex} if $f(\alpha\V{x} + (1-\alpha)\V{y}) \leq \alpha f(\V{x}) + (1-\alpha)f(\V{y})$
for all $\alpha \in (0,1)$ and $\V{x}, \V{y}$ in its domain. A function $f$ is {\em strictly convex} if the inequality is strict. 
\begin{proof}
The existence and uniqueness of a global minimizer $\M{U}^\star$ are immediate consequences of the coerciveness and strict convexity of $F_{\gamma}(\M{U})$.
$\square$
\end{proof}


\begin{proposition}[Continuity]
The solution $\M{U}^\star$ of (1) is jointly continuous in $(\M{X},\gamma, \M{W},\Mtilde{W})$.
\end{proposition}

\begin{proof}
Without loss of generality, we can absorb the regularization parameter $\gamma$ into the weights $\V{w} = (\vec(\M{W})\Tra,\vec(\Mtilde{W})\Tra)\Tra \in \Real^{\frac{p(p-1)}{2} + \frac{n(n-1)}{2}}$. Thus, we can check to see if the solution $\M{U}^\star$ is continuous in the variable $\V{\zeta} = (\vec(\M{X})\Tra,\V{w}\Tra)\Tra$. It is easy to verify that the following function is jointly continuous in $\M{U}$ and $\V{\zeta}$
\begin{eqnarray}
f(\M{U},\V{\zeta}) & = & \frac{1}{2} \lVert \M{X} - \M{U} \rVert_{\text{F}}^2 + J_{\V{W}}(\M{U}),
\end{eqnarray}
where
\begin{eqnarray}
J_{\V{W}}(\M{U}) & = & \frac{1}{\sqrt{p}} \Omega_{\M{w}}(\M{U}) + \frac{1}{\sqrt{n}}\Omega_{\Mtilde{w}}(\M{U}\Tra)
\end{eqnarray}
is a convex function of $\M{U}$ that is continuous in $\V{W}$. Let
\begin{eqnarray}
\M{U}^\star(\V{\zeta}) & = & \underset{\M{U}}{\arg\min}\; f(\M{U},\V{\zeta}).
\end{eqnarray}

We proceed with a proof by contradiction. Suppose $\M{U}^\star(\V{\zeta})$ is not continuous at a point $\V{\zeta}$. Then there exists an $\epsilon > 0$ and a sequence $\{\Vn{\zeta}{m}\}$ converging to $\M{\zeta}$ such that $\lVert \Mn{U}{m} - \M{U}^\star(\V{\zeta}) \rVert_{\text{F}} \geq \epsilon$ for all $m$ where
\begin{eqnarray}
\Mn{U}{m} & = & \underset{\M{U}}{\arg\min}\; f(\M{U},\Vn{\zeta}{m}).
\end{eqnarray}
Note that since $f(\M{U},\V{\zeta})$ is strongly convex in $\M{U}$, the minimizers $\Mn{U}{m}$ and $\M{U}^\star(\V{\zeta})$ exist and are unique. Without loss of generality we can assume
$\lVert \Vn{\zeta}{m} - \V{\zeta} \rVert_{\text{F}} \leq 1$.
This fact will be used later in proving the boundedness of the sequence $\Mn{U}{m}$.

Fix an arbitrary point $\Mtilde{U}$. If $\Mn{U}{m}$ is a bounded sequence then we can pass to a convergent subsequence with limit $\Mbar{U}$. Note that $f(\Mn{U}{m}, \Vn{\zeta}{m}) \leq f(\Mtilde{U}, \Vn{\zeta}{m})$ for all $m$. Since $f$ is continuous in $(\M{U},\V{\zeta})$, taking limits gives us the inequality
\begin{eqnarray}
f(\Mbar{U}, \V{\zeta}) & \leq & f(\Mtilde{U}, \V{\zeta}).
\end{eqnarray}
Since $\Mtilde{U}$ was selected arbitrarily, it follows that $\Mbar{U} = \M{U}^\star(\V{\zeta})$, which is a contradiction. It only remains for us to show that the sequence $\Mn{U}{m}$ is bounded.

Consider the function
\begin{eqnarray}
g(\M{U}) & = & \underset{\Vtilde{\zeta} : \lVert \Vtilde{\zeta}- \V{\zeta} \rVert_{\text{F}} \leq 1}{\sup}\;
\frac{1}{2} \lVert \Mtilde{X} - \M{U} \rVert_{\text{F}}^2 + J_\Vtilde{W}(\M{U}).
\end{eqnarray}
Note that $g$ is convex, since it is the point-wise supremum of a collection of convex functions. Since $f(\M{U}, \Vn{\zeta}{m}) \leq g(\M{U})$ and $f$ is strongly convex in $\M{U}$, it follows that $g(\M{U})$ is also strongly convex and therefore has a unique global minimizer $\M{U}^*$ such that $g(\M{U}^*) < \infty$. It also follows that
\begin{eqnarray}
\label{eq:ineqA}
f(\Mn{U}{m},\Vn{\zeta}{m}) & \leq & f(\M{U}^*,\Vn{\zeta}{m}) \amp \leq \amp g(\M{U}^*)
\end{eqnarray}
for all $m$. By the reverse triangle inequality it follows that
\begin{eqnarray}
\label{eq:ineqB}
\frac{1}{2} \left (
 \lVert \Mn{U}{m} \rVert_{\text{F}} - \lVert \Mn{X}{m} \rVert_{\text{F}} 
\right )^2 \amp \leq \amp \frac{1}{2} \lVert \Mn{U}{m} - \Mn{X}{m} \rVert_{\text{F}}^2 \amp \leq \amp
f(\Mn{U}{m},\Vn{\zeta}{m}).
\end{eqnarray}
Combining the inequalities in \Eqn{ineqA} and \Eqn{ineqB}, we arrive at the conclusion that
\begin{eqnarray}
\frac{1}{2} \left (
 \lVert \Mn{U}{m} \rVert_{\text{F}} - \lVert \Mn{X}{m} \rVert_{\text{F}} 
\right )^2 \amp \leq \amp g(\M{U}^*),
\end{eqnarray}
for all $m$. Suppose the sequence $\Mn{U}{m}$ is unbounded, namely $\lVert \Mn{U}{m} \rVert_{\text{F}} \rightarrow \infty$. 
But since $\Mn{X}{m}$ converges to $\M{X}$, the left hand side must diverge. Thus, we arrive at a contradiction if $\Mn{U}{m}$ is unbounded.
$\square$
\end{proof}


\begin{proposition}[Zeroes of the fusion penalty]
Under Assumption 1, \\ $J(\M{U}) = 0$ if and only if $\M{U} = c\V{1}\V{1}\Tra$ for some $c \in \Real$.
\end{proposition}

\begin{proof}
We first show that
\begin{eqnarray}
\Omega_{\M{W}}(\M{U}) & = & \sum_{i < j} \VE{w}{ij} \lVert \M{U}_{\cdot i} - \M{U}_{\cdot i} \rVert_{\text{F}}
\end{eqnarray}
is positive if and only if $\M{U}_{\cdot i} = \M{U}_{\cdot j}$ for all $i < j$, namely all the columns of $\M{U}$ are the same. Clearly if the columns of $\M{U}$ are the same, then $\Omega_{\M{W}}(\M{U})$ is zero. Suppose that $\Omega_{\M{W}}(\M{U})$ is zero. Then it must that be $\M{U}_{\cdot i} = \M{U}_{\cdot j}$ for every $\VE{w}{ij} > 0$. Consider a pair $(i,j)$ such that $\VE{w}{ij} = 0$. By Assumption 1, there exists a path $i \rightarrow k \rightarrow \cdots \rightarrow l \rightarrow j$ along which the weights are positive.  Let $w$ denote the smallest weight along this path, namely $w = \min \{\VE{w}{ik}, \ldots, \VE{w}{lj}\}$. By the triangle inequality
\begin{eqnarray}
\lVert \M{U}_{\cdot i} - \M{U}_{\cdot j} \rVert_{\text{F}} & \leq & \lVert \M{U}_{\cdot i} - \M{U}_{\cdot k} \rVert_{\text{F}}
+ \cdots + \lVert \M{U}_{\cdot l} - \M{U}_{\cdot j} \rVert_{\text{F}}.
\end{eqnarray}
We can then conclude that
\begin{eqnarray}
w\lVert \M{U}_{\cdot i} - \M{U}_{\cdot j} \rVert_{\text{F}} & \leq & \Omega_{\M{W}}(\M{U}) \amp = \amp 0.
\end{eqnarray}
It follows that $\M{U}_{\cdot i} = \M{U}_{\cdot j}$, since $w$ is positive. By a similar argument it follows that
\begin{eqnarray}
\Omega_{\Mtilde{W}}(\M{U}) & = & \sum_{i < j} \tilde{w}_{ij} \lVert \M{U}_{i \cdot} - \M{U}_{j\cdot} \rVert_{\text{F}}
\end{eqnarray}
is zero if and only if $\M{U}_{i \cdot} = \M{U}_{j \cdot}$ for all $i < j$, or in other words if the rows of $\M{U}$ are all the same. Thus, $J_\M{W}(\M{U}) = 0$ if and only if $\M{U}$ is a constant matrix.
\end{proof}

\label{sec:coalesce}

\begin{proposition}[Coalescence]
\label{prop:coalesce}
Under Assumption 1, $F_{\gamma}(\M{U})$ is minimized by the grand mean $\Mbar{\M{x}}$ for $\gamma$ sufficiently large.
\end{proposition}

\begin{proof}
We will show that there is a $\gamma_{\max}$ such that for all $\gamma \geq \gamma_{\max}$, the grand mean matrix $\Mbar{X}$ is the unique global minimizer to the primal objective \Eqn{biclust_objective_function}. We will certify that $\Mbar{X}$ is the solution to the primal problem by showing that the optimal value of a dual problem, which lower bounds the primal, equals $F_\gamma(\Mbar{X})$.

Throughout the proof, we will work with the vectorization of matrices, namely the vector obtained by stacking the columns of a matrix on top of each other. We denote the vectorization of a matrix $\M{X}$ by its corresponding bold lower case, namely $\V{x} = \vec(\M{X})$. Thus, we will construct a dual to the following representation of the primal problem,
\begin{eqnarray}
\label{eq:primal}
\underset{\V{u}}{\text{minimize}}\; F_\gamma(\V{u}) = \frac{1}{2} \lVert \V{x} - \V{U} \rVert_2^2 + \gamma J(\V{u}).
\end{eqnarray}
In order to rewrite the penalty $J$ in terms of the vector $\V{u}$, we use the identity $\vec(\M{M}\M{N}\M{P}) = (\M{P}\Tra \Kron \M{M})\vec(\M{N})$ where $\Kron$ denotes the Kronecker product between two matrices. Thus,
\begin{eqnarray}
J(\V{u}) & = & \sum_{i < j} \VE{w}{ij} \lVert \M{A}_{ij} \V{u} \rVert_2 + \sum_{i < j} \tilde{w}_{ij} \lVert \Mtilde{A}_{ij} \V{u} \rVert_2,
\end{eqnarray}
where
\begin{eqnarray}
\M{A}_{ij} & = & (\V{e}_i - \V{e}_j)\Tra \Kron \M{I}, \\
\Mtilde{A}_{ij} & = & \M{I} \Kron (\V{e}_i - \V{e}_j)\Tra,
\end{eqnarray}
and $\V{e}_i$ is the $i$th standard basis vector. To keep things notationally simpler, we have absorbed the normalizations by $\sqrt{p}$ and $\sqrt{n}$ into the weights $\VE{w}{ij}$ and $\tilde{w}_{ij}$.

We first introduce some notation in order to write the relevant dual problem to the primal problem (\ref{eq:primal}). Note that the column weights $\VE{w}{ij}$ can be identified with a column graph of $n$ nodes, where there is an edge between the $i$th and $j$th node if and only if $\VE{w}{ij} > 0$. The row weights $\tilde{w}_{ij}$ can also be identified with an analogous row graph of $p$ nodes.
Let $\mathcal{E}_c$ and $\mathcal{E}_r$ denote the sets of edges in the column and row graphs, and let $\lvert \mathcal{E}_c \rvert$ and $\lvert \mathcal{E}_r \rvert$ denote their respective cardinalities. The edge-incidence matrix of the column graph $\M{\Phi}_c \in \Real^{\lvert \mathcal{E}_c \rvert \times n}$ encodes its connectivity and is defined as
\begin{eqnarray}
\ME{\phi}{c,li} = \begin{cases}
1 & \text{If node $i$ is the head of edge $l$,} \\
-1 & \text{If node $i$ is the tail of edge $l$,} \\
0 & \text{otherwise.}
\end{cases}
\end{eqnarray}
The row edge-incidence matrix $\M{\Phi}_r \in \Real^{\lvert \mathcal{E}_r \rvert \times p}$ is defined similarly.

We begin deriving the dual problem by recalling that norms possess a variational representation in terms of their dual norms, namely
\begin{eqnarray}
\lVert \V{Y} \rVert & = & \underset{ \lVert \V{Z} \rVert_\dagger \leq 1 }{\max} \langle \V{Z}, \V{Y} \rangle,
\end{eqnarray}
where $\lVert \cdot \rVert_\dagger$ is the dual norm of $\lVert \cdot \rVert$. Using this fact and working through some tedious algebra, we can rewrite the penalty term in \Eqn{primal} compactly as
\begin{eqnarray}
\gamma J(\V{U}) & = & \underset{\V{v} \in C_\gamma}{\max}\; \left \langle \M{A}\Tra \V{v}, \V{u} \right \rangle.
\end{eqnarray}
The vector $\V{v}$ is the concatenation of several vectors, namely
\begin{eqnarray}
\V{v} = (\V{v}_1\Tra,\ldots, \V{v}_{\mathcal{E}_c}\Tra, \Vtilde{v}_1\Tra, \ldots, \Vtilde{v}_{\mathcal{E}_r}\Tra)\Tra,
\end{eqnarray}
where $\V{v}_l \in \Real^{n}, \Vtilde{v}_l \in \Real^p$. The matrix $\M{A}$ can be expressed in terms of the row and column edge-incidence matrices, namely
\begin{eqnarray}
\M{A} & = & \begin{pmatrix}
\M{\Phi}_c \Kron \M{I} \\
\M{I} \Kron \M{\Phi}_r \\
\end{pmatrix}.
\end{eqnarray}
Finally, the constraints on the vector $\V{v}$ are encoded in the set
$C_\gamma = \{ \V{v} : \lVert \V{V}_{l} \rVert_{2} \leq \VE{w}{l}\gamma$ and $\lVert \Vtilde{V}_{l} \rVert_2 \leq \tilde{w}_{l} \gamma\}$.

Thus, the primal problem \Eqn{primal} can be expressed as the following saddle point problem
\begin{eqnarray}
\underset{\V{v} \in C_\gamma}{\max}\;
\underset{\V{U}}{\min}\; \frac{1}{2} \lVert \V{x} - \V{u} \rVert_2^2 + \left \langle \M{A}\Tra \V{v}, \V{u} \right \rangle.
\end{eqnarray}
By performing the minimization with respect to $\V{u}$, we obtain a dual maximization problem that provides a lower bound on the primal objective
\begin{eqnarray}
\underset{\V{v} \in C_\gamma}{\max}\;
-\frac{1}{2} \lVert \M{A}\Tra\V{v} \rVert_2^2 + \langle \V{v}, \M{A}\V{x} \rangle.
\end{eqnarray}

For sufficiently large $\gamma$, the solution to the dual maximization problem coincides with the solution to the unconstrained maximization problem
\begin{eqnarray}
\underset{\V{v}}{\max}\;
-\frac{1}{2} \lVert \M{A}\Tra\V{v} \rVert_2^2 + \langle \V{v}, \M{A}\V{x} \rangle,
\end{eqnarray}
whose solution is $\V{v}^\star = \left (\M{A}\M{A}\Tra \right )^\dagger\M{A}\V{x}$.  Plugging $\V{v}^\star$ into the dual objective gives an optimal value of
\begin{eqnarray}
\frac{1}{2} \lVert \M{A}\Tra\left (\M{A}\M{A}\Tra \right )^\dagger \M{A} \V{x} \rVert_2^2 ,
\end{eqnarray}
which we rewrite as
\begin{eqnarray}
\frac{1}{2} \lVert \V{x} - \left [\M{I} - \M{A}\Tra\left (\M{A}\M{A}\Tra \right )^\dagger \M{A} \right ]\V{x} \rVert_2^2.
\end{eqnarray}
Note that  $\left [\M{I} - \M{A}\Tra\left (\M{A}\M{A}\Tra \right )^\dagger \M{A} \right ]$
is the projection onto the orthogonal complement of the column space of $\M{A}\Tra$, which is equivalent to the null space or kernel of $\M{A}$, denoted Ker$(\M{A})$. We will show shortly that Ker($\M{A}$) is the span of the all ones vector. Therefore, $\left [\M{I} - \M{A}\Tra\left (\M{A}\M{A}\Tra \right )^\dagger \M{A} \right ]\V{x} = \frac{1}{np}\langle \V{x}, \V{1} \rangle \V{1}$.

Before showing that Ker($\M{A}$) is the span of $\V{1}$, we note that the smallest $\gamma$ such that $\V{v}^\star \in C_\gamma$ is an upper bound on $\gamma_{\max}$. 

%

We now argue that Ker($\M{A}$) is the span of $\V{1} \in \Real^{np}$.
We rely on the following fact: If $\M{\Phi}$ is an incidence matrix of a connected graph with $n$ vertices, then the rank of $\M{\Phi}$ is $n-1$
(See Theorem 7.2 in Chapter 7 of \cite{Deo1974}). According to Assumption 1 in the paper, the column and row graphs are connected; it follows that $\M{\Phi}_c \in \{-1,0,1\}^{\lvert \mathcal{E}_c \rvert \times n}$ has rank $n-1$ and $\M{\Phi}_r \in \{-1,0,1\}^{\lvert \mathcal{E}_r \rvert \times p}$ has rank $p-1$. It follows then that Ker($\M{\Phi}_c$) and Ker($\M{\Phi}_r$) have dimension one. Furthermore, since each row of $\M{\Phi}_c$ and $\M{\Phi}_r$ has one $1$ and one $-1$, it follows that $\V{1} \in$ Ker($\M{\Phi}_c$) $\subset \Real^n$, and likewise $\V{1} \in$ Ker($\M{\Phi}_r$) $\subset \Real^p$. A vector $\V{z} \in $Ker($\M{A}$) if and only if $\V{z} \in$ Ker($\M{\Phi}_c \Kron \M{I}$) $\cap$ Ker($\M{I} \Kron \M{\Phi}_r$).

Recall that if the singular values of a matrix $\M{A}$ are $\sigma_{\M{A},i}$ and the singular values of a matrix $\M{B}$ are $\sigma_{\M{B},j}$, then the singular values of their Kronecker product $\M{A} \Kron \M{B}$ are $\sigma_{\M{A},i}\sigma_{\M{B},j}$. It follows then that the rank of $\M{A} \Kron \M{B}$ is the product of the ranks of $\M{A}$ and $\M{B}$.

The above rank property of Kronecker products of matrices implies that the dimension of Ker($\M{\Phi}_c \Kron \M{I})$ equals $p$ and the dimension of Ker($\M{I} \Kron \M{\Phi}_r)$ equals $n$. It is easy to see then that the linearly independent set of vectors $\{\V{1} \Kron \V{e}_1, \ldots,
\V{1} \Kron \V{e}_p\}$, where $\V{1} \in \Real^n$ and $\V{e}_i \in \Real^p$, forms a basis for Ker($\M{\Phi}_c \Kron \M{I}$). Likewise, the linearly independent set of vectors $\{ \V{e}_1 \Kron \V{1} , \ldots,
\V{e}_n  \Kron \V{1}\}$, where $\V{1} \in \Real^p$ and $\V{e}_i \in \Real^n$, forms a basis for Ker($\M{I} \Kron \M{\Phi}_r$).

Take an element from Ker($\M{\Phi}_c \Kron \M{I}$), namely $\V{1} \Kron \V{a}$, where $\V{1} \in \Real^n$ and $\V{a} \in \Real^p$. We will show that in order for $\V{1} \Kron \V{a} \in$ Ker($\M{i} \Kron \M{\Phi}_r$), $\V{a}$ must be a multiple of 1. Consider the relevant matrix-vector product 
\begin{eqnarray}
(\M{I} \Kron \M{\Phi}_r) (\V{1} \Kron \V{a}) \amp = \amp (\M{I} \Kron \M{\Phi}_r) \vec(\V{a}\V{1}\Tra) \amp = \amp\M{\Phi}_r \V{a}\V{1}\Tra,
\end{eqnarray}
where we again used the fact that $\vec(\M{M}\M{N}\M{P}) = (\M{P}\Tra \Kron \M{M})\vec(\M{N})$ and the fact that $\vec(\V{b}\V{c}\Tra) = \V{c}\Kron\V{b}$.
Note that $(\M{I} \Kron \M{\Phi}_r) (\V{1} \Kron \V{a}) = \V{0}$ if and only if $\M{\Phi}_r\V{a} = \V{0}$. But
the only way for $\M{\Phi}_r\V{a}$ to be zero is for $\V{a} = c\V{1}$ for some $c \in \Real$. A similar argument shows that the only non-trivial vector in Ker($\M{I} \Kron \M{\Phi}_r$) that also belongs to Ker($\M{\Phi}_c \Kron \M{I}$) is $\tilde{c}\V{1}$ for $\tilde{c} \in \Real$. Thus, we have shown that Ker($\M{A}$) is the span of $\V{1}$. $\square$
\end{proof}

\section*{Web Appendix C. DLPA, COBRA, and a Proof of Proposition 4.1}
\label{sec:DLPA}

We give expanded technical treatment of DLPA and COBRA as well as results described in Section 4. We begin by reviewing some basic concepts convex analysis \citep{BauCom2008,ComPes2011}. Recall that the domain of a convex function $f$ is the set of $\V{x}$ such that $f(\V{x}) < \infty$. For $\sigma > 0$ the mapping
\begin{eqnarray}
\prox_{\sigma f}(\V{u}) & = & \underset{\V{v}}{\arg\min}\;\left[\sigma f(\V{v})+ \frac{1}{2} \lVert \V{u} - \V{v} \rVert_2^2 \right]
\end{eqnarray}
is called the proximal map of the function $f(\V{v})$. The proximal map exists and is unique whenever the function $f(\V{v})$ is convex and lower-semicontinuous. 
Norms and semi-norms satisfy these conditions,
and for many norms of interest the proximal map can be evaluated by either an explicit formula or an efficient algorithm. For example, the proximal map for the $\ell_1$-norm is the ubiquitous element-wise soft-thresholding operator, namely the $l$th element of the proximal mapping is given by
\begin{eqnarray}
\left [\prox_{\sigma \lVert \cdot \rVert_1}(\V{u}) \right]_l & = & \left [ 1 - \frac{\sigma}{| \VE{u}{l} |} \right ]_+ \VE{u}{l}.
\end{eqnarray}

Closer inspection of \Eqn{biclust_objective_function} shows that we seek the proximal mapping of the sum of two lower-semicontinuous, convex functions, namely
\begin{eqnarray}
\prox_{f + g}(\M{U}) & = & \underset{\M{U}}{\arg\min}\; \frac{1}{2} \lVert \M{X} - \M{U} \rVert_{\text{F}}^2 + f(\M{U}) + g(\M{U}),
\end{eqnarray}
where $f(\M{U}) = \gamma\Omega_{\M{W}}(\M{U})$ and $g(\M{U}) = \gamma\Omega_{\Mtilde{W}}(\M{U}\Tra)$.

This problem is reminiscent of the classic problem of finding the projection of a point onto the intersection of two nonempty and closed convex sets. Indeed, it is the problem when the functions $f$ and $g$ are respectively the indicator functions of two nonempty and closed convex sets $A$ and $B$. Then we can pose the problem of finding the projection of $\M{X}$ onto the set $A \cap B$ as the optimization problem
\begin{eqnarray}
\underset{\M{U}}{\min}\; \frac{1}{2} \lVert \M{X} - \M{U} \rVert_{\text{F}}^2 + \delta_A(\M{U}) + \delta_B(\M{U}),
\end{eqnarray}
where $\delta_A$ is the set indicator function, which is 0 for all $\M{U} \in A$ and $\infty$ for all $\M{U} \not\in A$.

von Neumann's alternating projection method provides an iterative solution when the two sets $A$ and $B$ are vector subspaces \citep{Deu1992}. His strategy was subsequently generalized by Dykstra to closed convex cones in Euclidean spaces \citep{Dyk1983} and generalized further by Boyle and Dykstra to the intersection of convex sets in Hilbert spaces \citep{BoyDyk1986}. Finally, \cite{BauCom2008} derived a Dykstra-like proximal algorithm that iteratively solves for the desired proximal mapping of the sum of two convex functions \citep{BauCom2008,ComPes2011}, which we describe next.

Let $f$ and $g$ be lower-semicontinuous convex functions on $\Real^n$, with $\dom f \cap \dom g \not = \emptyset$, and let $\V{x} \in \Real^n$. Bauschke and Combettes' algorithm, shown in \Alg{DLPA},
iteratively solves the following problem
\begin{eqnarray}
\label{eq:dlpa}
\underset{\V{u} \in \Real^n}{\min}\; \frac{1}{2} \lVert \V{u} - \V{x} \rVert_2^2 + f(\V{u}) + g(\V{u})
\end{eqnarray}
and is guaranteed to converge.
\begin{theorem}[Proposition 5.3 in \cite{ComPes2011}]
\label{thm:dlpa}
\Alg{DLPA} converges to the solution of \Eqn{dlpa}.
\end{theorem}
\setcounter{algorithm}{2}
\begin{algorithm}[t]
Set $\V{u}_0 = \V{x}, \V{p}_0 = \V{0}, \V{q} = \V{0}$ for $m=0, 1, \ldots$
\begin{algorithmic}[0]
  \caption{Dykstra-Like Proximal Algorithm (DLPA)}
  \label{alg:DLPA}
\Repeat
\State $\V{y}_{m} = \prox_g(\V{u}_m + \V{p}_m)$
\State $\V{p}_{m+1} = \V{u}_m + \V{p}_m - \V{y}_m$
\State $\V{u}_{m+1} = \prox_f(\V{y}_m + \V{q}_m)$
\State $\V{q}_{m+1} = \V{y}_m + \V{q}_m - \V{u}_{m+1}$
\Until{convergence}
\end{algorithmic}
\end{algorithm}
Setting $f(\M{U}) = \gamma\Omega_{\M{W}}(\M{U})$ and $g(\M{U}) = \gamma\Omega_{\Mtilde{W}}(\M{U})$ in \Alg{DLPA} yields COBRA outlined in Algorithm 1 in the paper. Consequently, the convergence of COBRA (Proposition 4.1) follows immediately from \Thm{dlpa}, since $\Omega_{\M{W}}(\M{U})$ and $\Omega_{\Mtilde{W}}(\M{U})$ are both continuous convex functions over all of $\Real^{np}$.

\section*{Web Appendix D. Majorization-Minimization (MM) algorithms and a Proof of Proposition 5.1}
\label{sec:mm_proof}

Recall that in the model selection problem we seek the minimizer of the following validation objective
\begin{eqnarray}
\label{eq:validation}
\tilde{F}_\gamma(\M{U}) & = &
\frac{1}{2} \lVert \mathcal{P}_{\Theta^c}(\M{x}) - \mathcal{P}_{\Theta^c}(\M{U}) \rVert_{\text{F}}^2 + \gamma J(\M{U}).
\end{eqnarray}
In this appendix, we elaborate on how to extend COBRA via a Majorization-Minimization (MM) algorithm to handle missing data in order to solve \Eqn{validation}.

We begin with a brief review of MM algorithms. The basic strategy behind an MM algorithm is to convert a hard optimization problem into a sequence of simpler ones. The MM principle requires majorizing the objective function $f(\V{u})$ by a surrogate function $g(\V{u} \mid \Vtilde{u})$ anchored at the current point $\Vtilde{u}$.  Majorization is a combination of the tangency condition $g(\Vtilde{u} \mid \Vtilde{u}) =  f(\Vtilde{u})$ and the domination condition $g(\V{u} \mid \Vtilde{u})  \geq f(\V{u})$ for all $\V{u} \in \Real^n$.  The associated MM algorithm is defined by the iterates $\Vn{u}{m+1} := \underset{\V{u}}{\arg \min}\; g(\V{u} \mid \Vn{u}{m})$. It is straightforward to verify that the MM iterates generate a descent algorithm driving the objective function downhill, namely that $f(\Vn{u}{m+1}) \leq f(\Vn{u}{m})$ for all $m$.

Returning to our original problem,  we observe that the following quadratic function of $\M{U}$ is always nonnegative
\begin{eqnarray}
\label{eq:quad}
\frac{1}{2} \sum_{(i,j) \in \Theta} (\ME{u}{ij} - \tilde{u}_{ij})^2 & \geq & 0,
\end{eqnarray}
and that the inequality becomes equality when $\M{U} = \Mtilde{U}$. Adding the quadratic function in \Eqn{quad}
 to $\tilde{F}_\gamma(\M{U})$ gives us the following function
\begin{eqnarray}
g(\M{U} \mid \tilde{\M{U}}) & = & \tilde{F}_\gamma(\M{U}) + \frac{1}{2} \sum_{(i,j) \in \Theta} (\ME{u}{ij} - \tilde{u}_{ij})^2 \\
& = & \frac{1}{2} \left [ \sum_{(i,j) \in \Theta^c} (\ME{x}{ij} - \ME{u}{ij})^2 + \sum_{(i,j) \in \Theta} (\ME{U}{ij} - \tilde{u}_{ij})^2 \right ] + \gamma J(\M{U})\\
& = & \frac{1}{2} \lVert \M{M} - \M{U} \rVert_{\text{F}}^2 + \gamma J(\M{U}),
\end{eqnarray}
where $\M{M} = \mathcal{P}_{\Theta^c}(\M{X}) + \mathcal{P}_{\Theta}(\tilde{\M{U}})$. 
The function $g(\M{U} \mid \Mtilde{U})$ majorizes $\tilde{F}_\gamma(\M{U})$ at the point $\Mtilde{U}$, since
$g(\M{U} \mid \Mtilde{U}) \geq \tilde{F}_\gamma(\M{U})$ for all $\M{U}$ and $g(\Mtilde{U} \mid \Mtilde{U}) = \tilde{F}_\gamma(\Mtilde{U})$. Minimizing the majorization $g(\M{U} \mid \Mtilde{U})$ can be accomplished by invoking COBRA on the complete matrix $\M{M}$. Alternating between updating the majorization and applying COBRA to minimize the new majorization yields Algorithm~2 in the paper.

Having derived the majorization being minimized in Algorithm~2, we are almost ready to prove Proposition 5.1.
We need one more ingredient. The convergence theory of monotonically decreasing algorithms, like the MM algorithm, hinges on the properties of the map $\psi(\M{U})$ which returns the next iterate given the last iterate. For easy reference, we state a simple version of Meyer's monotone convergence theorem \citep{Meyer1976} which is the key ingredient in proving convergence in our setting.
\begin{theorem}\label{thm:MM_limit_points}
  Let $f(\M{U})$ be a continuous function on a compact domain $S$ and
   $\psi(\M{U})$ be a continuous map from $S$ into $S$ satisfying
 $f(\psi(\M{U})) < f(\M{U})$ for all $\M{U} \in S$ with $\psi(\M{U}) \neq \M{U}$.
Then all limit points are fixed points of $\psi(\M{U})$.
\end{theorem}
We now prove Proposition 5.1.
\begin{proof}
We first use the above theorem to establish that the iterates of the MM algorithm tend towards the fixed points of the corresponding map, $\psi(\Mtilde{U}) = \arg\min_{\M{U}} g(\M{U} \mid \Mtilde{U})$. Fix an arbitrary starting guess $\Mn{U}{0}$. Set $S = \{ \M{U} : g(\M{U} \mid \Mn{U}{0} ) \leq F_\gamma(\Mn{U}{0}) \}$. Since $g(\M{U} \mid \Mn{U}{0})$ is continuous and coercive in $\M{U}$, it follows that $S$ is compact. Since $g(\M{U} \mid \Mtilde{U})$ is strongly convex in $\M{U}$ it follows that if $\Mtilde{U}$ is not a fixed point then it is not the unique global minimizer of $g(\M{U} \mid \Mtilde{U})$ and therefore $F_\gamma(\psi(\Mtilde{U})) < F_\gamma(\Mtilde{U})$. By \Thm{MM_limit_points} the limit points of the sequence $\Mn{U}{n}$ are fixed points of $\psi(\M{U})$.

We argue that the fixed points of $\psi(\M{U})$ are global minimizers of the validation objective \Eqn{validation}. Note that the mapping $\Mtilde{U} \mapsto \psi(\Mtilde{U})$ is characterized by the condition
\begin{eqnarray}
\psi(\Mtilde{U}) - \mathcal{P}_{\Theta^c}(\M{X}) - \mathcal{P}_{\Theta}(\Mtilde{U}) & \in & \gamma\partial J(\psi(\Mtilde{U})).
\end{eqnarray}
If $\Mtilde{U}$ is a fixed point of $\psi$ then $\Mtilde{U} = \psi(\Mtilde{U})$, and the above optimality condition becomes
\begin{eqnarray}
\mathcal{P}_{\Theta^c}(\Mtilde{U}) - \mathcal{P}_{\Theta^c}(\M{X}) & \in & \gamma\partial J(\Mtilde{U}).
\end{eqnarray}
But this implies that $\Mtilde{U}$ is a global minimizer of the validation objective \Eqn{validation}. 

Putting everything together, we have that the limit points of the MM sequence $\Mn{U}{m}$ are global minimizers of the validation objective. $\square$
\end{proof}
Note that there might be infinitely many limit points. The set of limit points, however, must be contained in $S$ and is therefore bounded. It must also be convex and closed.

As a final remark, we point out that we obtain essentially the same MM algorithm, if we substitute the objective $\tilde{F}_\gamma(\M{U})$ with the objective
\begin{eqnarray}
\frac{1}{2} \lVert \mathcal{P}_{\Theta^c}(\M{x}) - \mathcal{P}_{\Theta^c}(\M{U}) \rVert_{\text{F}}^2 + \gamma \lVert \M{U} \rVert_*,
\end{eqnarray}
where $\lVert \M{U} \rVert_*$ denotes the nuclear norm of $\M{U}$. When we substitute $\lVert \M{U} \rVert_*$ for $J(\M{U})$, the resulting MM algorithm is the soft-impute algorithm of \cite{MazHasTib2010}.

\section*{Web Appendix E. Measures of Clustering Similarity}

\subsection*{Rand Index}
\label{sec:rand_index}

Consider the problem of clustering $q$ objects. Let $\mathcal{A} = \{A_1, \ldots, A_m\}$ and $\mathcal{B} = \{B_1, \ldots, B_n\}$ denote two partitions of the index set $\{1, \ldots, q\}$. The Rand index \citep{Ran1971} quantifies
how similar the partitions $\mathcal{A}$ and $\mathcal{B}$ are to each other by tallying up how often pairs of objects are similarly assigned and dividing this quantity by the total number of pairs of objects.  To compute the Rand Index,  we define four events:
\begin{equation}
\begin{split}
C_{ij} & =  \text{$\{ i,j \in A_k$ for some $k \}$} \\
D_{ij} & =  \text{$\{i,j \in B_l$ for some $l \}$} \\
E_{ij} & =  \text{$\{i \in A_k, j \in A_{k'}$ for $k \not = k'\}$} \\
F_{ij} & =  \text{$\{i \in B_l, j \in B_{l'}$ for $l \not = l'\}$}.
\end{split}
\end{equation}
The intersection $C_{ij} \cap D_{ij}$ denotes the event that $\mathcal{A}$ and $\mathcal{B}$ have assigned the pair of objects $i$ and $j$ similarly, namely under both partitions, $i$ and $j$ belong to the same cluster. 
The intersection $E_{ij} \cap F_{ij}$ denotes the event that $\mathcal{A}$ and $\mathcal{B}$ have assigned the pair of objects $i$ and $j$ similarly in an alternative sense, namely under both partitions, $i$ and $j$ are assigned to different clusters. The Rand index is given by the following ratio
\begin{eqnarray}
\frac{\sum_{i < j} I(C_{ij} \cap D_{ij}) + \sum_{i < j} I(E_{ij} \cap F_{ij})}{ {q \choose 2} },
\end{eqnarray}
where $I(Z)$ is 1 if event $Z$ occurs and 0 otherwise. By dividing by the total number of pairs, the Rand index takes on values between 0 (no agreement between $\mathcal{A}$ and $\mathcal{B}$) and 1 (perfect agreement between $\mathcal{A}$ and $\mathcal{B}$).

\subsection*{Adjusted Rand Index}
\label{sec:adjusted_rand_index}

\begin{table}[th]
\begin{tabular}{ c || c c c c || c }
& $\mathcal{B}_1$ & $\mathcal{B}_2$ & $\cdots$ & $\mathcal{B}_n$ & Row Marginal \\ \hline
$\mathcal{A}_1$ & $q_{11}$ & $q_{12}$ & $\cdots$ & $q_{1n}$ & $q_{1\cdot}$  \\ 
$\mathcal{A}_2$ & $q_{21}$ & $q_{22}$ & $\cdots$ & $q_{2n}$ & $q_{2\cdot}$  \\ 
$\vdots$ & $\vdots$ & $\vdots$ & $\ddots$ & $\vdots$ & $\vdots$  \\
$\mathcal{A}_m$ & $q_{m1}$ & $q_{m2}$ & $\cdots$ & $q_{mn}$ & $q_{m\cdot}$ \\
Column Marginal & $q_{\cdot 1}$ & $q_{\cdot 2}$ & $\cdots$ & $q_{\cdot n}$ & $q$ \\ \hline \\
\end{tabular}
\caption{\label{tab:adjusted_rand} 
Contingency table comparing two partitions $\mathcal{A} = \{A_1, \ldots, A_m\}$ and $\mathcal{B} = \{B_1, \ldots, B_n\}$ of the index set $\{1, \ldots, q\}$.}
\end{table}

While the Rand index is an intuitive and simple quantitative measure of assessing the similarity of two partitions, it does have some well known defects. For example, since the Rand index makes no probabilistic assumptions on the data, there are no guarantees on the behavior of the Rand index when assessing the similarity between two random partitions. 
To address these issues, \cite{HubAra1985} proposed the adjusted Rand index, which assumes a hypergeometric distribution in the null case when two random partitions are being compared. Note that this is a very strong assumption. Nonetheless, the benefit of making this assumption is that under the null case, the expected value of the adjusted Rand index is zero. When there is perfect agreement between the two partitions being compared the adjusted Rand index is 1. Unlike the Rand index, it is possible for the adjusted Rand index to take on negative values. We next review details on calculating the adjusted Rand index.

Again consider the problem of clustering $q$ objects. Let $\mathcal{A} = \{A_1, \ldots, A_m\}$ and $\mathcal{B} = \{B_1, \ldots, B_n\}$ denote two partitions of the index set $\{1, \ldots, q\}$. 
We illustrate how the score is computed using the contingency table shown in \Tab{adjusted_rand}. The $ij$th entry in the table $q_{ij}$ denotes the number of elements common to $\mathcal{A}_i$ and $\mathcal{B}_j$.
We denote the $i$th row marginal sum $q_{i\cdot} = \sum_{j} n_{ij}$ and the $j$th column marginal sum $q_{\cdot j} = \sum_{i} q_{ij}$. The adjusted Rand Index is given by the following ratio.

\begin{eqnarray}
\frac{\sum_{i=1}^m \sum_{j=1}^n {q_{ij} \choose 2} - \left [\sum_{i=1}^m {q_{i \cdot} \choose 2} \sum_{j=1}^n {q_{\cdot j} \choose 2}\right ]/{q \choose 2}}
{\frac{1}{2} \left [ \sum_{i=1}^m {q_{i \cdot} \choose 2} + \sum_{j=1}^n {q_{\cdot j} \choose 2} \right ] - \left [ \sum_{i=1}^m {q_{i \cdot} \choose 2} \sum_{j=1}^n {q_{\cdot j} \choose 2} \right ]/{q \choose 2}}.
\end{eqnarray}

\subsection*{Variation of Information}
\label{sec:variation_information}

We first need to define the entropy of a partition and the mutual information between two partitions in order to define the variation of information. As before suppose we have $q$ objects to cluster and two partitions $\mathcal{A}$ and $\mathcal{B}$. Let $\lvert \mathcal{A}_i \rvert$ denote the number of elements in the $i$th partition $\mathcal{A}_i$.
Entropy of a partition $\mathcal{A}$ is given by
\begin{eqnarray}
\mathcal{H}(\mathcal{A}) & = & \sum_{i=1}^m \frac{\lvert \mathcal{A}_i \rvert}{q}\log_2 \left(\frac{\lvert \mathcal{A}_i \rvert}{q} \right).
\end{eqnarray}
The mutual information between two partitions $\mathcal{A}$ and $\mathcal{B}$ is given by
\begin{eqnarray}
\mathcal{I}(\mathcal{A}, \mathcal{B}) & = & \sum_{i=1}^m \sum_{j=1}^n \frac{\lvert \mathcal{A}_i \cap \mathcal{B}_j \vert}{q} \log_2 \left ( \frac{\lvert \mathcal{A}_i \cap \mathcal{B}_j \rvert}{\lvert \mathcal{A}_i\rvert \lvert \mathcal{B}_j\rvert} \right ).
\end{eqnarray}
The variation of information between two clustering is given by
\begin{eqnarray}
\mathcal{VI}(\mathcal{A},\mathcal{B}) & = & \mathcal{H}(\mathcal{A}) + \mathcal{H}(\mathcal{B}) - 2\mathcal{I}(\mathcal{A},\mathcal{B}).
\end{eqnarray}

\section*{Web Appendix F. Non-Checkerboard Mean Structure}

In the checkerboard model, we have assumed that the mean structure is succinctly described by the cross-product of row and column partitions. This next example explores how COBRA performs when this assumption is violated. Consider the case where the observed data matrix $\M{X}$ is a noisy realization of 5 underlying biclusters $(a, b, c, d, e)$ that can be arranged as follows:
\begin{eqnarray}
\label{eq:ref2_1}
\M{X} & = & \begin{bmatrix}
\mu_a\M{1}_a & \mu_a\M{1}_a & \mu_d\M{1}_d \\
\mu_b\M{1}_b & \mu_c\M{1}_c & \mu_e\M{1}_e \\
\end{bmatrix} + \M{\mathcal{E}},
\end{eqnarray}
where $\M{1}_a$ is a matrix of ones, $\mu_a$ is the mean of the bicluster-$a$, and $\mathcal{E}$ is a matrix whose entries are independent draws from a Gaussian distribution. As noted by a referee, this is a scenario that is likely to occur in practice and consists of biclusters that violate the cross-product structure assumed by the checkerboard model. We simulated data $\ME{x}{ij} = \ME{\mu}{k} + \ME{\varepsilon}{ij}$ where $\mu := (\mu_a, \mu_b, \mu_c, \mu_d, \mu_e) = (1,0,0.25,-1,1.25)$ and $\ME{\varepsilon}{ij}$ are i.i.d.\@ draws from $N(0,0.1)$. Using the validation procedure to select the regularization parameter $\gamma$, COBRA identified 2 row partitions and 3 column partitions with bicluster means $\hat{\mu} = (1.000, 1.001, -0.002, 0.252, -0.999, 1.248)$. \Fig{non_checkerboard} shows the smoothed COBRA estimate. While COBRA cannot exactly identify the true bicluster structure, the true structure is readily identifiable from the COBRA output given that the two biclusters in the first row and first two columns have nearly identical estimated means. In short, by examining the estimated biclusters in conjunction with their estimated means, COBRA can potentially identify the correct biclusters even when the checkerboard assumption is violated.

\begin{figure}
\centering
\includegraphics[scale=0.5]{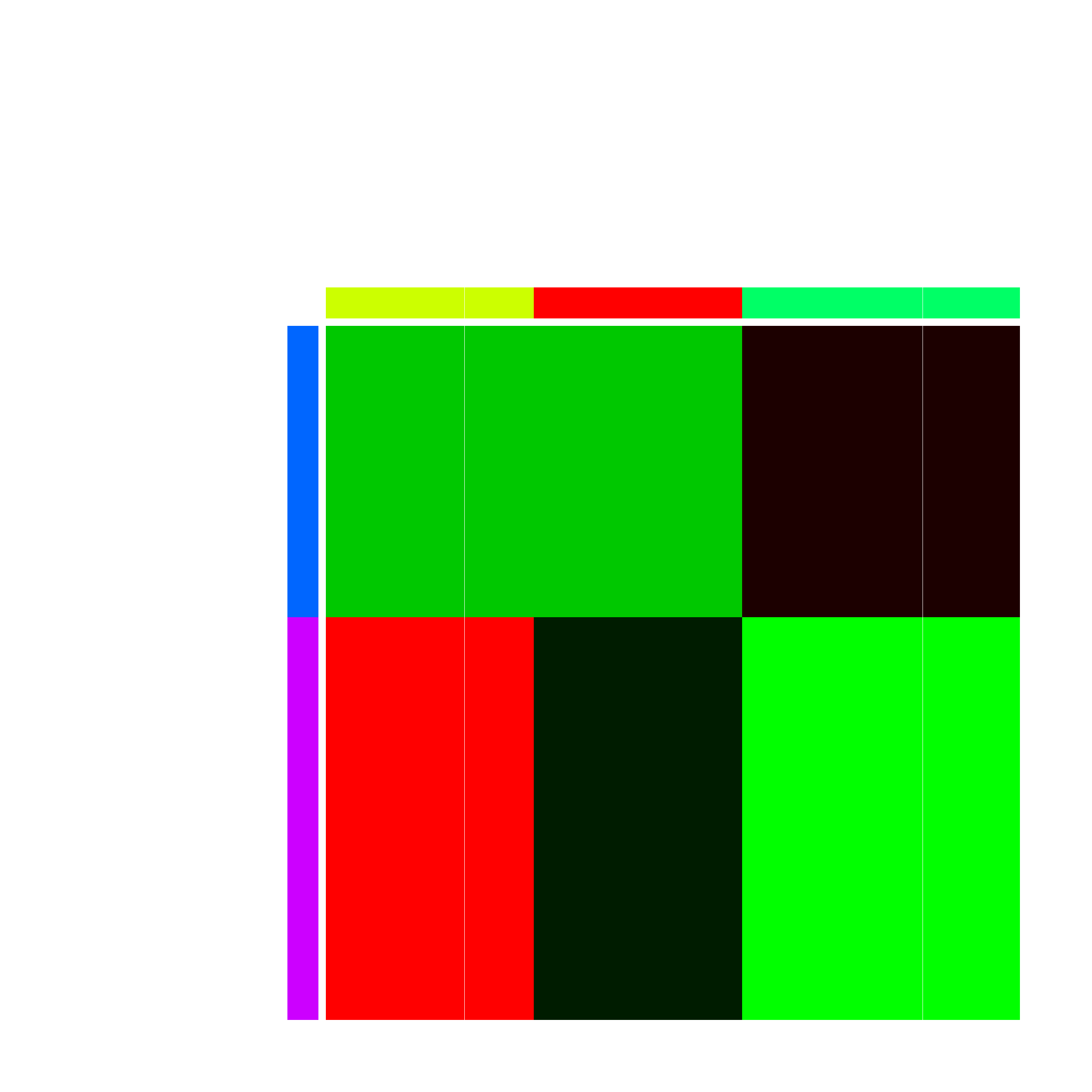}
\caption{Non-Checkerboard mean structure. While COBRA cannot exactly identify the true bicluster structure, the true biclustering structure is readily identifiable from the COBRA estimate.}
\label{fig:non_checkerboard}
\end{figure}

\section*{Web Appendix G. COBRA Refinements}

\subsection*{Adaptive COBRA}

We describe a natural extension of the adaptive Lasso \citep{Zou2006} to the convex biclustering problem. The adaptive COBRA applies the COBRA method twice. Let $\M{U}^\star$ denote the first COBRA solution. Note this first application of COBRA includes the model selection step detailed in Section 5 in the main paper. So, $\M{U}^\star$ corresponds to the smoothed estimate at the chosen $\gamma$ to minimize the hold-out error. We then recompute the weights treating $\M{U}^\star$ as the data and using the same sparse Gaussian kernel weights procedure detailed in Web Appendix A. We then perform a second round of the COBRA method using these new weights on the original data $\M{X}$. This two step procedure shrinks together more strongly similar columns (rows) than the original COBRA method, mimicking the effect of reweighting in the adaptive Lasso.

\subsection*{Thresholded COBRA}

Before we can describe how the thresholded COBRA works, we need to review how clustering assignments are made in convex clustering. Recall that COBRA alternates between applying convex clustering on the rows and columns of the data matrix $\M{X}$.  To streamline the discussion, we focus on how convex clustering obtains column clusters; row clusters are obtained analogously. As noted at the end of Section 4, hard clustering assignments are trivially obtained from variables employed in the splitting method introduced in \cite{ChiLan2015}. We elaborate on this comment here. To cluster the columns of $\M{X}$, we solve the following minimization problem.
\begin{eqnarray}
\underset{\M{U}}{\min}\; \frac{1}{2} \lVert \M{X} - \M{U} \rvert_{\text{F}}^2 + \gamma \sum_{i < j} w_{ij} \lVert \M{U}_{\cdot i} - \M{U}_{\cdot j} \rVert_2,
\end{eqnarray}
where $w_{ij}$ are weights that differentially penalize pairwise differences between columns based on their degree of similarity and the tuning parameter $\gamma$ trades off the emphasis between the data fit and the smoothness of the solution. The AMA method proposed by \cite{ChiLan2015} solves the following equivalent minimization problem.
\begin{eqnarray}
\underset{\M{U}, \V{V}_{1,2}, \ldots, \V{V}_{n-1,n}}{\min}\; \frac{1}{2} \lVert \M{X} - \M{U} \rvert_{\text{F}}^2 + \gamma \sum_{i < j} w_{ij} \lVert \V{V}_{i,j} \rVert_2,
\end{eqnarray}
subject to $\V{v}_{i,j} = \M{U}_{\cdot i} - \M{U}_{\cdot j}$ for all $i < j$. We have introduced a dummy variable $\V{v}_{i,j}$ that is the difference between the $i$th and $j$th columns of $\M{U}$. The AMA method iteratively applies group-wise softthresholding to send $\V{v}_{i,j}$ vectors with small magnitude to zero. Thus, cluster assignments are made as follows. If $\V{v}_{i,j} = \V{0}$, then the $i$th and $j$th columns are put in the same group. An analogous set of dummy vectors are obtained after applying convex clustering on the rows of $\M{X}$.

We are now ready to describe a natural extension of the thresholded Lasso \citep{Meinshausen2009} to the convex biclustering problem. As with the adaptive COBRA, the thresholded COBRA performs a postprocessing step that groups together column centroids (row centroids) that are almost but not exactly identical. In solving the COBRA optimization problem we obtain a set of vectors associated with the column differences $\V{v}_{i,j}$. We compute a second set of column difference vectors $\Vtilde{V}_{i,j}$ from the column difference vectors $\V{V}_{i,j}$ as follows.
\begin{eqnarray}
\Vtilde{v}_{i,j} & = & \begin{cases}
\V{v}_{i,j} & \text{if $\lVert \V{v}_{i,j} \rVert_2 \geq \tau$} \\
\V{0} & \text{otherwise}
\end{cases}.
\end{eqnarray}
Under the thresholded COBRA, if $\Vtilde{v}_{i,j} = \V{0}$, then the $i$th and $j$th columns are assigned to the same column group.
In short, the postprocessing consists of hard thresholding of the column difference vectors $\V{V}_{i,j}$. The parameter $\tau$ controls how aggressive the hard-thresholding is. A natural question is how to set $\tau$. In the case of sparse linear regression $\tau$ should be on the order of the noise \citep{Meinshausen2009}. This is typically estimated in practice using the standard deviation in the residuals. In this work we choose $\tau$ to be a fraction of the standard deviation of the 2-norms of the vectors $\V{v}_{i,j}$. To be conservative, we chose $\tau$ to be 1/4 of this standard deviation.

\section*{Web Appendix H. Additional Stability Experiments}

We repeat the stability experiments on the lung cancer data (Section 7 of the main paper) using three other biclustering methods with software available in R.
These methods are (i) the iterative signature algorithm \citep{BerIhm2003} implemented in the package {\tt isa2}, (ii) the sparse SVD method \citep{LeeSheHua2010,SilKaiKop2011} implemented in the package {\tt s4vd}, and (iii) Plaid models \citep{LazOwe2002} implemented in the package {\tt biclust}. 
These three methods are popular SVD-based approaches that seek overlapping biclusters. 
This is in contrast to the methods compared in the main paper that assume an underlying checkerboard mean structure. The reader should keep this difference in mind since the biclustering output of these three methods are not directly comparable to the biclustering output of the methods considered in the main paper. Nonetheless, despite these differences, it is possible to assess the stability of these methods in the same manner that we evaluated the stability for the methods that assume a checkerboard pattern. We include these results for completeness here. All parameters were selected according to the default methods in the R packages. 

As before, we restrict our attention to the 150 genes with the highest variance. We first apply the biclustering methods on the original data to obtain baseline biclusterings. We then add iid $N(0,\sigma^2)$, noise where $\sigma = 0.5, 1.0, 1.5$ to create a perturbed data set on which to apply the same set of methods. We compute the RI, ARI, and VI between the baseline clustering and the one obtained on the perturbed data.  \Tab{stability} shows the average RI, ARI, and VI of 50 replicates as well as run times. All three methods are relatively fast and take run times between that of COBRA and the clustered dendrogram with dynamic tree cutting. The Plaid model, however, clearly exhibits the best stability among these approaches that seek overlapping biclusters. For convenience, we have included \Tab{stability_checkerboard} which shows the stability results of the COBRA variants, DCT, and spBC that are presented in Table~2 of the main paper.

\begin{table}[th]
\begin{tabular}{ l c c c c }
& $\sigma$ & ISA & sparse SVD & Plaid \\ \hline
RI & 0.5 & 0.731 & 0.957 & 0.983 \\ 
& 1.0 & 0.595 & 0.955 & 0.963 \\
& 1.5 & 0.528 & 0.936 & 0.944 \\ \hline
ARI & 0.5 & 0.462 & 0.567 & 0.880 \\ 
& 1.0 & 0.191 & 0.515 & 0.655 \\
& 1.5 & 0.054 & 0.372 & 0.353 \\ \hline
VI & 0.5 & 1.243 & 0.213 & 0.097 \\ 
& 1.0 & 1.624 & 0.222 & 0.189 \\
& 1.5 & 1.746 & 0.273 & 0.231 \\ \hline
time (sec) & 0.5 & 2.60 & 9.07 & 2.65 \\ 
& 1.0 & 2.94 & 10.63 & 3.02 \\
& 1.5 & 2.89 & 10.22 & 2.95 \\ \hline \\
\end{tabular}
\caption{\label{tab:stability} 
Stability and reproducibility of biclusterings in lung cancer microarray data. The ISA, sparse SVD, and Plaid biclustering methods are applied to the lung cancer data to obtain baseline biclusterings. We then perturb the data by adding iid\@ $N(0,\sigma^2)$ noise where $\sigma = 0.5$ (Small Pert.), 1.0 (Medium Pert.), 1.5 (Large Pert.).} 
\end{table}

\begin{table}[th]
\begin{tabular}{ l c c c c c c }
& $\sigma$ & COBRA & COBRA (A) & COBRA (T) & DCT & spBC \\ \hline
RI & 0.5 & 0.984 & 0.992 & 0.959 & 0.979 & 0.974 \\ 
& 1.0 & 0.981 & 0.990 & 0.944 & 0.974 & 0.965 \\
& 1.5 & 0.973 & 0.989 & 0.896 & 0.973 & 0.936 \\ \hline
ARI & 0.5 & 0.350 & 0.788 & 0.813 & 0.530 & 0.642 \\ 
& 1.0 & 0.233 & 0.686 & 0.766 & 0.439 & 0.544 \\
& 1.5 & 0.201 & 0.667 & 0.644 & 0.340 & 0.397 \\ \hline
VI & 0.5 & 1.924 & 0.882 & 0.776 & 2.120 & 1.568 \\ 
& 1.0 & 2.380 & 1.276 & 0.962 & 2.769 & 2.174 \\
& 1.5 & 2.721 & 1.312 & 1.320 & 3.505 & 2.915 \\ \hline
time (sec) & 0.5 & 15.44 & 23.46 & 15.65 & 0.07 & 151.59 \\ 
& 1.0 & 25.21 & 34.51 & 25.53 & 0.11 & 197.00 \\
& 1.5 & 18.18 & 26.43 & 18.50 & 0.12 & 207.88 \\ \hline \\
\end{tabular}
\caption{\label{tab:stability_checkerboard} 
Stability and reproducibility of biclusterings in lung cancer microarray data. COBRA variants, the clustered dendrogram with dynamic tree cutting, and sparse Biclustering are applied to the lung cancer data to obtain baseline biclusterings. We then perturb the data by adding iid\@ $N(0,\sigma^2)$ noise where $\sigma = 0.5$ (Small Pert.), 1.0 (Medium Pert.), 1.5 (Large Pert.).} 
\end{table}

\bibliographystyle{imsart-nameyear}
\bibliography{biclustering}

\end{document}